\documentclass[pdflatex,iicol,sn-mathphys,Numbered]{sn-jnl}
\usepackage{longtable}
\usepackage{amsmath}
\usepackage{amsthm}
\usepackage{amssymb}
\usepackage{silence}
\usepackage{caption}
\usepackage{subcaption}
\usepackage[dvipsnames,table]{xcolor}
\usepackage{amsfonts}
\usepackage{mathtools}
\usepackage{multirow}
\usepackage{url}
\newcommand{\burl}[1]{\url{#1}}
\usepackage{xspace}
\usepackage{type1cm}
\usepackage{natbib}
\usepackage[title]{appendix}
\usepackage{xcolor}
\usepackage{textcomp}
\usepackage{manyfoot}
\usepackage{algorithm}
\usepackage{algorithmicx}
\usepackage{algpseudocode}
\usepackage{listings}
\usepackage{colortbl}

\def\braket#1{\mathinner{\langle{#1}\rangle}}
\newcommand{\bra}[1]{\left\langle #1\right|}
\newcommand{\ket}[1]{\left|#1\right\rangle}
\newcommand{\QDAGer}{\textup{\textsc{QDAGer}}\xspace}
\newcommand{\defeq}{\vcentcolon=}

\definecolor{rankBlue}{RGB}{210,229,255}
\definecolor{rankGreen}{RGB}{217,242,208}
\definecolor{rankOrange}{RGB}{255,229,204}
\newcommand{\bestcell}[1]{\cellcolor{rankBlue}#1}
\newcommand{\secondcell}[1]{\cellcolor{rankGreen}#1}
\newcommand{\thirdcell}[1]{\cellcolor{rankOrange}#1}

\newcommand{\gap}{\mathrm{gap}}

\newtheorem{theorem}{Theorem}
\newtheorem{lemma}[theorem]{Lemma}

\newtheorem{proposition}[theorem]{Proposition}
\newtheorem{corollary}[theorem]{Corollary}
\def\<{\langle}
\def\>{\rangle}
\theoremstyle{definition}

\newtheorem*{remark}{Remark}
\usepackage{booktabs}
\usepackage{graphicx}

\begin{document}
\title[Transverse-Field Ising Model for Graph Learning]{On the Expressive Power of the Transverse-Field Ising Model for Graph Learning}

\newcommand{\orcidid}[1]{\href{https://orcid.org/#1}{\textsuperscript{\textcolor[HTML]{A6CE39}{\textbf{iD}}}}}

\author*[1]{\fnm{Mehdi} \sur{Djellabi}\orcidid{0000-0003-3537-7855}}\email{mehdi.djellabi@pasqal.com}
\author[1]{\fnm{Louis-Paul} \sur{Henry}\orcidid{0000-0002-2013-5215}}\email{louis-paul.henry@pasqal.com}

\affil*[1]{\orgname{Pasqal}, \orgaddress{\street{24 avenue Émile Baudot}, \city{Palaiseau}, \postcode{91120}, \country{France}}}
\maketitle

\begin{abstract}
{}We study the quantum evolution induced by graph-indexed Ising Hamiltonians as a source of structural signal for graph learning. Graph automorphisms preserve symmetries of the Hamiltonian, and these symmetries constrain the quantum evolution in a way that turns time-dependent local measurements into informative probes of graph structure. Leveraging this idea, we introduce \QDAGer, a quantum-inspired graph-pair Transformer that injects quantum-dynamical features from time series of node occupations and connected two-point correlators directly into the attention mechanism. We apply \QDAGer to learning Graph Edit Distance (GED), an NP-hard similarity measure, using either a direct permutation-invariant embedding discrepancy or an alignment-based surrogate loss. Experiments on multiple GED benchmarks under different edit cost settings show that the proposed dynamical features provide a stronger inductive bias than classical structural alternatives under the same training protocol. In addition, we report ablations where the dynamical signal is replaced by standard random-walk and heat-kernel features while keeping the architecture fixed, highlighting that the gain comes from the injected dynamics rather than model capacity alone.
\end{abstract}
\section{Introduction}
Learning on graphs is central to modern machine learning because many real-world systems are inherently relational, including molecules, proteins, code, social networks, and physical interaction networks~\cite{wu2021comprehensive,bronstein2021geometric,gilmer2017neural}. The key challenge is to build representations that capture meaningful structure while respecting permutation invariance to node relabeling~\cite{maron2019invariant}. This challenge is especially pronounced in graph comparison tasks, where success depends on detecting fine-grained differences while remaining robust to symmetries and repeated patterns.

A substantial body of work has studied the expressive power of graph models. Message-passing GNNs are effective in many applications, but their ability to distinguish non-isomorphic graphs is fundamentally limited by local aggregation schemes, which can fail when global structure must be resolved~\cite{xu2018how,morris2019wlgnn}. Graph Transformers provide a complementary mechanism: self-attention can relate distant nodes in a single layer and can, in principle, model global dependencies~\cite{dwivedi2020generalization,ying2021graphormer,vaswani2017attention}. In practice, however, attention over graphs does not automatically yield strong structural understanding. Without carefully designed positional encodings, Transformers may struggle to represent symmetries, and can behave like powerful models that lack a reliable notion of graph geometry~\cite{kreuzer2021rethinking,lim2023sign,dwivedi2022graph}. This motivates a recurring question: which structural signals should be injected into attention to obtain representations that are both informative and permutation-consistent?

In this work, we propose a \emph{dynamical} and \emph{tunable} structural signal derived from the quantum evolution of a graph-indexed transverse-field Ising Hamiltonian. Given an undirected graph \(\mathcal{G}=(\mathcal{V},\mathcal{E})\), we consider an Ising-type Hamiltonian whose interaction terms are indexed by edges, and whose transverse field and longitudinal field act on nodes.
Evolving a simple initial state and measuring local observables over time yields time series
that reflect how excitations, \textit{i.e.} nodes in the occupied (excited) state, and the corresponding correlations propagate on the graph. From a graph learning
viewpoint, these measurements act as multi-scale probes of structure: varying the evolution
time changes the effective range of interactions captured by the signal, while the
nonlinearity of the dynamics enriches the representational family beyond static encodings.
This perspective is related to recent quantum-feature approaches for graph learning based
on neutral-atom dynamics and graph-dependent observables~\cite{PhysRevA.104.032416,QuantPE,hetero_graph}.
Whereas these works typically use quantum-derived quantities as \emph{static} kernels or
pre-computed graph features consumed by a downstream classifier, and quantum-walk
approaches~\cite{kasture2025multiparticlequantumwalksdistinguishing} target graph
\emph{distinguishability} rather than a learnable metric, we instead inject time-resolved
node occupations and connected two-point correlators \emph{directly into the attention
mechanism} of a graph Transformer, so that the dynamical signal shapes a trainable,
alignment-aware comparison rather than a fixed feature map.

A key motivation for using this signal is its tight connection to graph symmetries. Graph automorphisms preserve symmetries of the Hamiltonian, meaning that permuting nodes by an automorphism leaves the dynamics unchanged. Consequently, the evolution and the resulting observables are naturally permutation-consistent and symmetry-aware. Intuitively, this provides a mechanism for producing features that respond to ``structural roles'' rather than to arbitrary labels: nodes that are equivalent under the automorphism group are constrained to share the same dynamical behavior, while nodes in different structural roles can exhibit distinct time-dependent signatures. This makes quantum-dynamical observables attractive candidates for injecting structural bias into attention, in a way that is operationally tied to the input graph.

We leverage these ideas to introduce \QDAGer (\textbf{Q}uantum \textbf{D}ynamics-based \textbf{A}ttention \textbf{G}raph transform\textbf{er}), a graph-pair Transformer designed for graph comparison. For each graph $\mathcal{G}$, we emulate its quantum evolution and record node-level time series of occupations $\langle \hat n_u(t)\rangle$ and pair-level time series of connected correlators $\langle \hat n_u\hat n_v(t)\rangle - \langle \hat n_u(t)\rangle\langle \hat n_v(t)\rangle$ at a discrete set of sampling times. Given a training pair $(\mathcal{G}_1,\mathcal{G}_2)$, \QDAGer forms combinations of these signals to produce attention logits over both nodes and node pairs. 
The model then performs coupled updates: pair representations are updated through gated interactions driven by node-derived queries and keys, and node representations aggregate both node values and attention-weighted information coming from updated pair features. In this way, \QDAGer integrates dynamical information at both the node and pair level directly into the attention mechanism, rather than treating it as a separate positional encoding.

We evaluate \QDAGer on learning Graph Edit Distance (GED), a widely used similarity measure defined as the minimum-cost sequence of node and edge edit operations required to transform one graph into another. GED is NP-hard to compute exactly~\cite{GEDComplexity}, making it a challenging benchmark that requires capturing both global structure and fine-grained local differences while remaining invariant to node permutations. We consider two training strategies. The first is a direct permutation-invariant embedding discrepancy (a ``blindfold'' setting) that regresses a distance between the two sets of learned embeddings to the true GED. The second is an alignment-based surrogate loss, following the neural set divergence framework of~\cite{JainNeuralGED}, in which a learned soft alignment produces doubly-stochastic transport plans over nodes and pairs, yielding a differentiable proxy for edit-based comparison.

Empirically, we evaluate on multiple GED benchmarks under both uniform and non-uniform edit cost regimes. To isolate the contribution of the proposed signal, we include ablations that replace the quantum-dynamical inputs with classical structural alternatives such as relative random-walk and heat-kernel features, while keeping the architecture and training protocol fixed. Across datasets where sufficient training pairs are available, \QDAGer with quantum-dynamical attention achieves consistently stronger performance than these classical variants, suggesting that the quantum dynamics provide a useful inductive bias for graph comparison with Transformers.

\section{Graph-Indexed Hamiltonian Dynamics}
\label{sec:theory}
\subsection{Mapping graphs to Hamiltonians}
\label{subsec:graph_to_ham}
We analyze graph-indexed Ising Hamiltonian dynamics as a symmetry-consistent mechanism for probing graph structure. Throughout this section, we consider an undirected, unweighted graph
$
\mathcal{G}=(\mathcal{V},\mathcal{E})
$
with \(|\mathcal{V}|=N\).

Each vertex \(u\in\mathcal V\) is associated with a two-level quantum system. The two local states are denoted by \(|0\rangle\) and \(|1\rangle\), and the corresponding two-dimensional Hilbert space is
$
\mathrm{span}\{|0\rangle,|1\rangle\}.
$
We denote by \(\widehat\sigma_u^x\) and \(\widehat\sigma_u^z\) the Pauli operators acting on the qubit indexed by node \(u\). With the convention that \(|0\rangle\) and \(|1\rangle\) are eigenvectors of \(\widehat\sigma^z\) with respective eigenvalues \(1\) and \(-1\), we define the occupation-number operator
$
\widehat n_u
=
\frac12\bigl(\mathbb I-\widehat\sigma_u^z\bigr).
$

The Hamiltonian used throughout this work is
\begin{equation}
\label{eq:ham}
\widehat H_{\mathcal{G}}
=
\sum_{(u,v)\in\mathcal{E}}
\widehat n_u\widehat n_v
+
\sum_{u\in\mathcal{V}}
\Omega\,\widehat\sigma_u^x
-
\delta\,\widehat n_u .
\end{equation}
The first term couples occupations along the edges of \(\mathcal G\), the second term is a transverse field, and the third term is a longitudinal field written in occupation-number form. In this work, the parameters \(\Omega\) and \(\delta\) are kept constant in time.

Equation~\eqref{eq:ham} is an occupation-number parametrization of a transverse-field Ising Hamiltonian indexed on the graph \(\mathcal G\). Indeed, using
$
\widehat n_u=\frac12(\mathbb I-\widehat\sigma_u^z),
$
we obtain :
$
\widehat n_u\widehat n_v
=
\frac14
\left(
\mathbb I
-
\widehat\sigma_u^z
-
\widehat\sigma_v^z
+
\widehat\sigma_u^z\widehat\sigma_v^z
\right).
$
Consequently, up to an additive scalar multiple of the identity, Eq.~\eqref{eq:ham} can be
rewritten as
\begin{multline}
\widehat H_{\mathcal{G}}
=
\Omega\sum_{u\in\mathcal V}\widehat\sigma_u^x
+
\frac14\sum_{(u,v)\in\mathcal E}
\widehat\sigma_u^z\widehat\sigma_v^z
\\
+
\sum_{u\in\mathcal V}
\left(
\frac{\delta}{2}
-
\frac{\deg(u)}{4}
\right)
\widehat\sigma_u^z
+
\mathrm{const.}
\end{multline}
Where $\deg(u)$ denotes the degree, \textit{i.e.} the number of neighbors of node $u$. The additive constant only contributes a global phase to the quantum evolution and therefore does not affect the expectation-value dynamics of the observables used below. 

The occupation-number form of Eq.~\eqref{eq:ham} is also naturally motivated by neutral-atom Rydberg physics. In physical neutral-atom devices, the states \(|0\rangle\) and \(|1\rangle\) can be interpreted as the absence or presence of a Rydberg excitation, and the corresponding van der Waals Hamiltonian is of the form
\begin{equation}
\label{eq:phys_motivation}
\begin{split}
\widehat H_{\mathrm{phys}}(t)
=&
\sum_{u=1}^{N}
\frac{\Omega(t)}{2}\widehat\sigma_u^x
-
\sum_{u=1}^{N}
\delta(t)\widehat n_u\\
&+
\sum_{1\leq u<v\leq N}
\frac{C_6}{r_{uv}^6}
\widehat n_u\widehat n_v,
\end{split}
\end{equation}
where \(r_{uv}\) is the distance between atoms \(u\) and \(v\), and \(C_6\) depends on the chosen Rydberg level. 

For graph classes admitting suitable geometric embeddings, such as unit-disk graphs, blockade physics can approximate edge-indexed interaction patterns by placing adjacent vertices within an effective blockade radius and non-adjacent vertices farther apart~\cite{pichler2018quantum,Henriet2020quantumcomputing}. However, the theoretical and learning constructions in this work do not rely on such a geometric realization. We treat Eq.~\eqref{eq:ham} as an abstract Hamiltonian whose interaction graph is the input graph itself. This distinction is important because several graphs in the experimental datasets are not guaranteed to admit a direct embedding on currently available neutral-atom hardware. Consequently, all numerical results reported in this work are obtained through emulation of Eq.~\eqref{eq:ham}, rather than through execution on neutral-atom hardware. Accordingly, we position \QDAGer as a \emph{quantum-inspired} method: the Ising dynamics act as a physically motivated, tunable feature generator, and we make no claim of quantum computational advantage over classical methods.
\subsection{Graph and wavefunction symmetries}
\label{subsec:graph_automorphisms_wavefunction_symmetries}
Let
$
\Gamma:=\mathrm{Aut}(\mathcal{G})
$
be the automorphism group of \(\mathcal{G}\). Each \(g\in\Gamma\) acts on nodes by
\(u\mapsto g\star u\), and this action is represented on the Hilbert space by a unitary
\(U_g\) permuting qubit labels. Since \(g\) preserves the edge set, the Hamiltonian is
invariant under this action:
$
U_g\widehat H_{\mathcal{G}}U_g^\dagger
=
\widehat H_{\mathcal{G}},
\;\;
\forall g\in\Gamma.
$
Thus the graph automorphisms are also symmetries of the quantum dynamics.

As in~\cite{PhysRevA.104.032416}, we associate to \(\mathcal{G}\) the time-evolved
quantum state
$
|\psi_{\mathcal{G}}(t)\rangle
=
e^{-it\widehat H_{\mathcal{G}}}|\psi_0\rangle,
$
living in the \(2^N\)-dimensional Hilbert space \(\mathcal H\). By Maschke's theorem, the
representation of \(\Gamma\) on \(\mathcal H\) decomposes the latter into a direct sum of
irreducible \(\Gamma\)-invariant subspaces~\cite{fulton1991representation}. We focus on
the trivial symmetry sector, denoted \(\mathcal H_0\), consisting of states invariant under
all automorphisms:
$
U_g|\psi\rangle=|\psi\rangle,
\;\; g\in\Gamma.
$
Since \(\widehat H_{\mathcal G}\) commutes with every \(U_g\), the sector \(\mathcal H_0\)
is invariant under the dynamics. Therefore, if \(|\psi_0\rangle\in\mathcal H_0\), then
$
|\psi_{\mathcal{G}}(t)\rangle\in\mathcal H_0,
\;\;
\forall t\in\mathbb R.
$

We use
$
\bigl\{|\boldsymbol b\rangle
:=
\bigotimes_{u=1}^N |b_u\rangle
\bigr\}_{\boldsymbol b\in\{0,1\}^N}
$
as the computational basis of \(\mathcal H\). Each bitstring
\(\boldsymbol b\in\{0,1\}^N\) naturally corresponds to the subset
$
S_{\boldsymbol b}
:=
\{u\in\mathcal V : b_u=1\},
$
and therefore to the induced subgraph \(\mathcal G[S_{\boldsymbol b}]\). The automorphism
group acts on \(S_{\boldsymbol b}\) by
$
g\star S_{\boldsymbol b}
:=
\{g\star u:u\in S_{\boldsymbol b}\},
$
and hence also on bitstrings. In particular, the all-zero state
\(|0\rangle^{\otimes N}\) belongs to the trivial symmetry sector \(\mathcal H_0\).

In this sector, basis amplitudes are constant on automorphism orbits: if two bitstrings
\(\boldsymbol b\) and \(\boldsymbol b'\) belong to the same \(\Gamma\)-orbit, then their
time-evolving coefficients in \(|\psi_{\mathcal G}(t)\rangle\) are identical whenever the
initial state lies in \(\mathcal H_0\). This orbit structure is the mechanism through which
graph symmetries constrain the wavefunction dynamics. In particular, \(k\)-body
occupation observables can probe orbits of \(k\)-vertex subsets, and therefore reveal
symmetric substructures of the input graph. Detailed derivations of these coefficient
equalities and of the corresponding orbit basis are provided in
Appendix~\ref{app:basis_details_2}.
\subsection{Theoretical guarantees}
\label{subsec:problem_statement_theory}
The aim of this section is to capture the symmetry-distinguishing content encoded in the expectation-value dynamics of occupation observables. Rather than relying on a sequential procedure as in Weisfeiler--Leman-type algorithms, we study the full time-dependent signal corresponding to the quantum dynamics of Hamiltonians that are indexed on graphs. The main result below shows that, within the trivial symmetry sector, occupation observables associated with distinct subset orbits yield generically distinct expectation-value trajectories. The guarantee is stated for Haar-random initial quantum states in $\mathcal{H}_0$, which rules out (sparse) pathological initial configurations.

Let $S,S'\subseteq\mathcal V$ two subsets belonging to distinct orbits under the induced action of $\Gamma$ on $2^{\mathcal V}$. Define the occupation observables
\begin{equation}
\widehat{O}_S
	:=
	\prod_{u\in S}\widehat n_u
\;\; \text{and} \;\;
\widehat O_{S'}
	:=
	\prod_{u\in S'}\widehat n_u,
\end{equation}
with the convention that the empty product is the identity. Let $\mu_{\mathrm{Haar}}$ be the normalized Haar measure on $\mathbb S(\mathcal H_0)$, the unit sphere of $\mathcal H_0$, whose elements represent normalized $\Gamma$-invariant pure states up to global phase. Then one can state the following :

\begin{theorem}[Distinguishability in the trivial symmetry sector]
\label{thm:main}
For $|\psi_0\rangle\in\mathbb S(\mathcal H_0)$, define
\begin{multline}
f_{S,S'}(t)
:=
\big\langle \psi_0\big|
e^{it\widehat H_{\mathcal G}}
\big(\widehat O_S-\widehat O_{S'}\big)
e^{-it\widehat H_{\mathcal G}}
\big|\psi_0\big\rangle,
\\
t\in\mathbb R.
\end{multline}
Then, for $\mu_{\mathrm{Haar}}$-almost every $|\psi_0\rangle\in\mathbb S(\mathcal H_0)$:
$f_{S,S'}\not\equiv 0$.
\newline Consequently,
$\lambda_1\bigl(\{t\in\mathbb R:f_{S,S'}(t)=0\}\bigr)=0$,
where $\lambda_1$ denotes the one-dimensional Lebesgue measure.
\end{theorem}
 In other terms, this result states that for Haar-random initial states in $\mathcal H_0$, two occupation observables associated with distinct subset orbits have expectation-value trajectories that coincide only on a set of times that has zero Lebesgue measure. Importantly, the two observables remain distinct after compression to the dynamically accessible symmetry sector.

Note that Theorem~\ref{thm:main} is not strictly restricted for the Hamiltonian described in Eq.~\eqref{eq:ham}. Indeed the proof (\textit{cf.} Appendix~\ref{app:proof_main}) only requires finite dimensionality and automorphism invariance of the Hamiltonian, namely $[\widehat H_{\mathcal G},U_g]=0$ for every $g\in\Gamma$. Thus, this result is more generally a symmetry-sector statement for graph-indexed Hamiltonian dynamics. 

We now explain how this orbit-level distinguishability connects to graph isomorphism. The key point is that, for connected graphs, isomorphism can be witnessed inside the automorphism group of the disjoint union by a symmetry that mixes the two connected components. The following result describes when such component-mixing symmetries can occur.

Given two simple connected graphs $\mathcal{G}_1=(\mathcal{V}_1,\mathcal{E}_1)$ and $\mathcal{G}_2=(\mathcal{V}_2,\mathcal{E}_2)$ on disjoint vertex sets, we fix the notation
\begin{equation}
\label{eq:def_G12}
\mathcal{G}_{1,2}:=\mathcal{G}_1\cup \mathcal{G}_2,
\qquad
\Gamma_{1,2}:=\mathrm{Aut}(\mathcal{G}_{1,2}).
\end{equation}

\begin{theorem}[Automorphisms of disconnected unions,~\cite{Frucht_49}]
\label{thm:frucht}
Let $\mathcal{G}_1=(\mathcal{V}_1,\mathcal{E}_1)$ and $\mathcal{G}_2=(\mathcal{V}_2,\mathcal{E}_2)$ be two simple connected graphs with disjoint vertex sets, and let $\mathcal{G}_{1,2}$ be defined as in Eq.~\eqref{eq:def_G12}. If $\mathcal{G}_1\not\simeq \mathcal{G}_2$, then every automorphism of $\mathcal{G}_{1,2}$ preserves each connected component, and therefore
$
\mathrm{Aut}(\mathcal{G}_{1,2}) \simeq \mathrm{Aut}(\mathcal{G}_1)\times \mathrm{Aut}(\mathcal{G}_2).
$
If $\mathcal{G}_1\simeq \mathcal{G}_2$, then $\mathrm{Aut}(\mathcal{G}_{1,2})$ contains automorphisms exchanging the two connected components.
\end{theorem}

In our case, we will use an equivalent (and handy) orbit-level formulation of Theorem~\ref{thm:frucht}, with $\Gamma_{1,2}$ as in Eq.~\eqref{eq:def_G12}.
\begin{corollary}[Orbit-crossing criterion]
\label{cor:corolary_main}
Let $\mathcal{G}_1$ and $\mathcal{G}_2$ be two simple connected graphs on disjoint vertex sets. Then $\mathcal{G}_1\simeq\mathcal{G}_2$ if and only if there exists an orbit of $\Gamma_{1,2}$ containing at least one vertex from $\mathcal{V}_1$ and at least one vertex from $\mathcal{V}_2$.
\end{corollary}

Indeed, if $\mathcal{G}_1\simeq\mathcal{G}_2$ then Theorem~\ref{thm:frucht} provides an automorphism exchanging the components, so the orbit of any vertex intersects both components. Conversely, an orbit intersecting both components contains $v\in\mathcal{V}_1$ and $\sigma(v)\in\mathcal{V}_2$ for some $\sigma\in\Gamma_{1,2}$, which forces $\sigma$ to map the connected component $\mathcal{V}_1$ onto $\mathcal{V}_2$, hence yields an isomorphism $\mathcal{G}_1\simeq\mathcal{G}_2$.
The proof (\textit{cf.} Appendix~\ref{appx:cor}) follows directly from Theorem~\ref{thm:frucht} and the fact that automorphisms preserve connected components.

Because $\mathcal G_{1,2}:=\mathcal G_1\cup\mathcal G_2$ has vertex set $\mathcal V_{1,2}=\mathcal V_1\sqcup \mathcal V_2$ and edge set $\mathcal E_{1,2}=\mathcal E_1\sqcup \mathcal E_2$, the Hamiltonian in Eq.~\eqref{eq:ham} contains no interaction term coupling a qubit in $\mathcal V_1$ to a qubit in $\mathcal V_2$. Identifying the Hilbert space as $\mathcal H_{\mathcal G_{1,2}}\simeq \mathcal H_{\mathcal G_1}\otimes \mathcal H_{\mathcal G_2}$ (tensor product over qubits in $\mathcal V_1$ and $\mathcal V_2$), this yields the separable decomposition
\begin{equation}
\widehat H_{\mathcal G_1\cup\mathcal G_2}
=
\widehat H_{\mathcal G_1}\otimes I
+
I\otimes \widehat H_{\mathcal G_2},
\end{equation}
where the two summands act on disjoint sets of qubits and therefore commute. Consequently, the propagator factorizes as
$
e^{-it\widehat H_{\mathcal G_1\cup\mathcal G_2}}
=
e^{-it\widehat H_{\mathcal G_1}}\otimes e^{-it\widehat H_{\mathcal G_2}},
$
so the dynamics on the disjoint union factorize for product initializations 
\begin{multline}
\label{eq:initial_prod_state}
    |\psi_0\rangle=|\psi^{(1)}_0\rangle\otimes|\psi^{(2)}_0\rangle
    \;\;\text{where }
    \\
    |\psi_0^{(1)}\rangle\in \mathcal H_0(\mathcal{G}_1),\;\;
    |\psi_0^{(2)}\rangle\in \mathcal H_0(\mathcal{G}_2).
\end{multline}
Thus, vertex-level occupation dynamics can be compared independently across $\mathcal G_1$ and $\mathcal G_2$. Applying Theorem~\ref{thm:main} to singleton subsets yields a generic separation of cross-component vertex orbits whenever the two graphs are not isomorphic.

Let  $\mathcal{G}_1$, $\mathcal{G}_2$  be two non-isomorphic graphs, and $\mathcal{G}_{1,2}$ be defined as in Eq.~\eqref{eq:def_G12}. Assuming that the initial state on the disjoint union factorizes as in Eq.~\eqref{eq:initial_prod_state}, we have:

\begin{proposition}[Measurement complexity]
\label{thm:stat_test}
For $\mu_{\mathrm{Haar}}^{(1)}\otimes\mu_{\mathrm{Haar}}^{(2)}$-almost every such product initial state, and for almost every $t_0>0$, the gap at $t_0$ :
$
\mathrm{gap}(t_0)
:=
\min_{u\in \mathcal{V}_1,\;v\in \mathcal{V}_2}
\left|
\langle \widehat n_u(t_0)\rangle
-
\langle \widehat n_v(t_0)\rangle
\right|,
$
is strictly positive. Consequently, given a confidence level $1-\alpha$, the two graphs can be distinguished from vertex-occupation measurements using
$
O\!\left(
		\log\!\left(\frac{N}{\alpha}\right)
		\gap(t_0)^{-2}
	\right)
$
measurement shots where $N=|\mathcal V_1|+|\mathcal V_2|$.
\end{proposition}

Here, $\mu_{\mathrm{Haar}}^{(1)}$ and $\mu_{\mathrm{Haar}}^{(2)}$ denote the normalized Haar measures on the unit spheres of $\mathcal H_0(\mathcal G_1)$ and $\mathcal H_0(\mathcal G_2)$, respectively.
To justify the positivity of $\gap(t_0)$, note that for every cross-component pair $u\in\mathcal V_1$, $v\in\mathcal V_2$, Corollary~\ref{cor:corolary_main} implies that $\{u\}$ and $\{v\}$ belong to distinct $\Gamma_{1,2}$-orbits whenever $\mathcal G_1\not\simeq\mathcal G_2$. Theorem~\ref{thm:main} therefore implies that
$
t\mapsto
	\langle \widehat n_u(t)\rangle
	-
	\langle \widehat n_v(t)\rangle
$
is not identically zero for Haar-almost every initial state $|\psi_0\rangle \in \mathbb{S}(\mathcal{H}_0)$. Since there are only finitely many cross-component pairs, the union of their zero sets still has Lebesgue measure zero. Hence $\gap(t_0)>0$ for almost every $t_0>0$. The shot-complexity bound then follows from the same Hoeffding/union-bound argument detailed in Appendix~\ref{app:appendix_stat_test}.
\subsection{Empirical validation}
\label{subsec:empirical_validation}

Note that Theorem~\ref{thm:main} is a \emph{generic} distinguishability statement over initial states in the trivial symmetry sector: it guarantees separation for $\mu_{\mathrm{Haar}}$-almost every $|\psi_0\rangle\in\mathbb S(\mathcal H_0)$. In practice, however, preparing or explicitly sampling generic states in $\mathcal H_0$ is graph-dependent, since $\mathcal H_0$ is determined by the automorphism group $\Gamma$. By contrast, only a few simple choices are guaranteed to belong to the trivial symmetry sector for every graph, such as fully permutation-invariant product or superposition states, including $|0\rangle^{\otimes N}$ and $|1\rangle^{\otimes N}$.

Our pipeline uses the fixed deterministic initialization $|\psi_0\rangle=|0\rangle^{\otimes N}$. For the dynamics resulting from the system described by Eq.~\eqref{eq:ham}, this is the simplest computational-basis product state and is independent of the graph instance. This choice is also consistent with the physical motivation of Eq.~\eqref{eq:ham}, where $|0\rangle^{\otimes N}$ corresponds to the native product state with no excitations. Theorem~\ref{thm:main} does not by itself guarantee that this particular initialization is separating for arbitrary Hamiltonian parameters. The experiments below therefore complement the generic theory by testing whether, in the driven regime used in our implementation, the same orbit-induced separation effect is already visible from the specific initialization $|0\rangle^{\otimes N}$. We refer the reader to the end of Appendix~\ref{app:proof_main} for further discussion.
\subsubsection{Hard graphs separation.}
To showcase the expressive power of the Ising dynamics, we consider $\mathcal{G}_1$ the \((4,4)\) Rook’s graph~\cite{Hoffman64Rook}, and $\mathcal{G}_2$ the Shrikhande graph~\cite{Shrikhande}, a non-isomorphic strongly regular graph (srg) pair~\cite{bose1963strongly} with identical parameters \((16,6,2,2)\) that requires at least the 4-WL test for a deterministic distinction~\cite{ARVIND202042, kasture2025multiparticlequantumwalksdistinguishing}. Both graphs are moreover \emph{vertex-transitive}: their automorphism groups act transitively on the vertex set, such that all vertices form a single orbit. Combined with the automorphism-invariant initialization \(\ket{\psi_0}=\ket{0}^{\otimes N}\in\mathcal{H}_0\) (\textit{cf.} section~\ref{subsec:graph_automorphisms_wavefunction_symmetries}), this implies that all vertices of a given graph share the same one-body occupation trajectory \(\langle \widehat n_u(t)\rangle\). Comparing $\mathcal{G}_1$ and $\mathcal{G}_2$ therefore reduces from the \(16\times16\) cross-vertex comparisons \((u_1,u_2)\), with \(u_1\in\mathcal V_1\) and \(u_2\in\mathcal V_2\), to a single comparison between one representative vertex per graph.
\newline Figure~\ref{fig:rook_vs_shrikhande} shows that although both graphs share strong regularity and symmetries, their one-body occupation trajectories \(\langle \widehat n_u(t)\rangle\) differ under the Hamiltonian evolution from Eq.~\eqref{eq:ham}. We therefore compute the minimum number of measurement shots required to certify this separation at a prescribed statistical confidence (here \(95\%\)), \textit{i.e.}, the measurement protocol yields a probabilistic distinction whose reliability increases with the shot budget. This behavior is consistent with recent empirical evidence~\cite{QuantPE} showing that observables derived from quantum random walks and Ising-type dynamics can resolve difficult families of non-isomorphic strongly regular graphs. As a classical baseline, the \(k\)-dimensional Weisfeiler--Leman procedure (\(k\)-WL) colors all ordered \(k\)-tuples in \(\mathcal V^k\) (thus \(N^k\) states) and, in a naïve refinement round, aggregates over all substitutions in each coordinate at cost \(O(kN^{k+1})\) operations (up to polylogarithmic overhead for hashing/sorting multiset signatures~\cite{Leman2018THERO}). Since each non-stable round strictly refines a partition of \(V^k\) with at most \(N^k\) classes, stabilization occurs after at most \(O(N^k)\) rounds in the worst case, yielding \(T_{k\text{-WL}}(N)=O(kN^{2k+1})\) (often written \(\widetilde O(N^{2k+1})\)). For \(k=4\) this gives \(O(N^9)\). In our setting \(N=16\), a single refinement round entails \(\approx 4\cdot 16^5=4.19\times10^6\) tuple updates, and the pessimistic worst-case total reaches \(16^9=6.87\times10^{10}\), whereas practical instances typically stabilize in far fewer than \(N^4\) rounds.
\begin{figure*}[htbp]
\centering
\includegraphics[width=\linewidth]{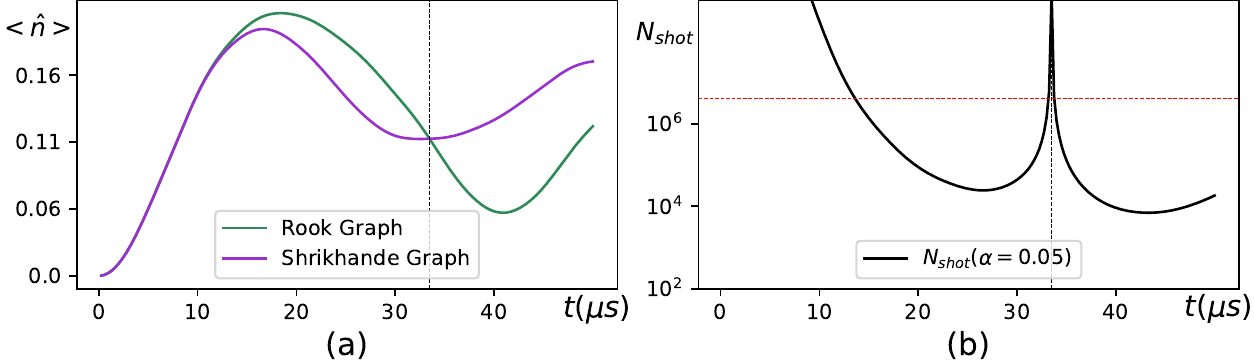}
\caption{
Rook versus Shrikhande graphs distinction. (a) One-body occupation dynamics \(\langle \hat n(t)\rangle\). Within each graph, the vertex-transitivity property ensures that all vertices share the same trajectory due to a single vertex orbit. (b) Number of shots required to distinguish the two graphs with success probability at least \(95\%\). The red dashed line indicates the lower bound for the time complexity scale of the exact classical counterpart (4-WL). Parameters: \(\ket{\psi_0}=\ket{0}^{\otimes N}\) for both graphs, \(\delta=0\) and \(\Omega=0.1\).
}
\label{fig:rook_vs_shrikhande}
\end{figure*}

\subsubsection{Study of larger instances.}
We now assess the separability of larger strongly regular graph families, namely $\operatorname{srg}(26,10,3,4)$ and $\operatorname{srg}(25,12,5,6)$, which contain respectively $10$ and $15$ non-isomorphic graphs. For each graph, the dynamics are initialized in the common all-zero state $\ket{0}^{\otimes N}$ and evolved under the graph-indexed Ising Hamiltonian with the same constant parameters as in the previous example: $\Omega=0.1$ and $\delta=0$.
Since graphs in these families are not vertex transitive, the comparison between two graphs cannot be reduced to a single representative node and must instead be performed at the level of the full node-resolved signature. We therefore compare the complete node-occupation trajectories modulo vertex relabeling.

All graphs within a given family share the same number of vertices $N$ (here $N=26$ and $N=25$ respectively), and their node occupations are sampled on a common discrete time grid $t_1<t_2<\dots<t_T$, where $T$ is the total number of timesteps. For a pair $(\mathcal G_i,\mathcal G_j)$ and vertices $u,v$, we define the per-vertex time-averaged mismatch
\begin{equation}
m_{uv}
:=
\frac{1}{T}\sum_{l=1}^{T}
\left|
\langle \hat n_u(t_l)\rangle_{\mathcal G_i}
-
\langle \hat n_v(t_l)\rangle_{\mathcal G_j}
\right|^2,
\label{eq:cost_matrix}
\end{equation}
and the subsequent permutation-invariant separation index:
\begin{equation}
\eta^{(1)}_{ij}
=
\min_{\pi\in S_{N}}
\left(
\frac{1}{N}\sum_{u=1}^{N} m_{u,\pi(u)}
\right)^{1/2},
\label{eq:RMS}
\end{equation}

where $S_{N}$ is the symmetric group acting on the $N$ vertices, so that the minimization ranges over all one-to-one relabelings of the vertices of $\mathcal G_j$ onto those of $\mathcal G_i$ (\textit{cf.} Appendix~\ref{appx:expe_details_1} for further details). The prefactor $1/(N\,T)$ averages the squared occupation mismatch over the $N$ nodes and over the $T$ sampled timesteps, so that $\eta^{(1)}_{ij}$ is a root-mean-square (RMS) distance. By construction, it vanishes if and only if the two graphs share identical one-point occupation trajectories up to a relabeling of their vertices. A strictly positive value therefore certifies that the two graphs are distinguished by the one-point dynamics alone.

Figure~\ref{fig:srg_separability} reports, for each family, the pair-averaged separation $\bar\eta^{(1)}$ together with its one-standard-deviation spread, and the worst-case separation $\min_{i<j}\eta^{(1)}_{ij}$, shown as a function of the elapsed physical time $t_m$ as the horizon grows over the $T$ time steps. Starting from a common initial state, all curves vanish at early times, before the occupations have had time to differentiate. As the evolution proceeds, the worst-case separation becomes and remains strictly positive, well above the numerical tolerance, for both families. This shows that every pair of non-isomorphic graphs in $\operatorname{srg}(26,10,3,4)$ and $\operatorname{srg}(25,12,5,6)$ is separated by the one-point occupation dynamics alone, thereby extending the separability results to these larger instances.
\begin{figure*}[htbp]
    \centering
    \includegraphics[width=\linewidth]{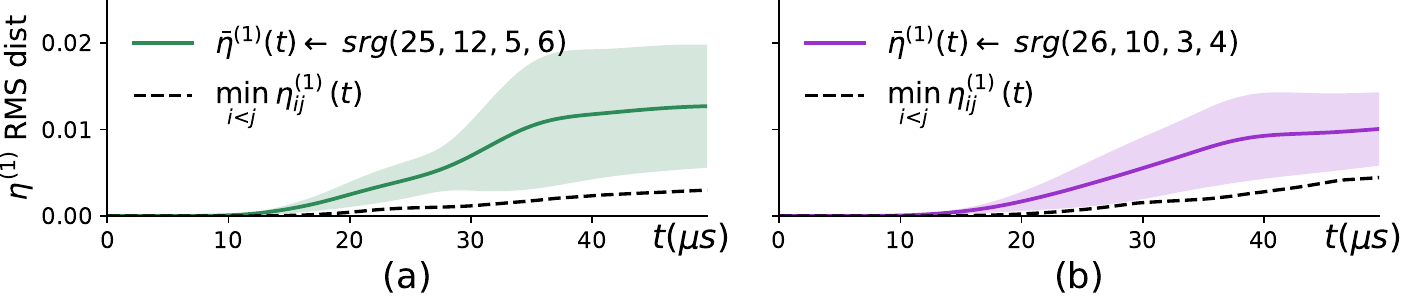}
\caption{Permutation-invariant separation index of strongly regular graph families $\operatorname{srg}(25,12,5,6)$ (a) and $\operatorname{srg}(26,10,3,4)$ (b), containing $15$ and $10$
non-isomorphic graphs respectively. The minimization over permutations $\pi\in S_N$ is solved exactly as a linear assignment problem between the node-occupation time traces of graph pairs, and $\eta^{(1)}_{ij}(t)$ is expressed as a root-mean-square (RMS) distance over nodes and time steps. All graphs are initialized in the common $\ket{0}^{\otimes N}$ state and evolved with constant parameters $\Omega=0.1$ and $\delta=0$. Solid lines show the pair-averaged
separation $\bar\eta^{(1)}(t)$ and shaded bands the corresponding
$\pm1$ standard deviation over graph pairs. Dashed black lines show the worst-case separation index $\min_{i<j}\eta^{(1)}_{ij}(t)$, where $\eta^{(1)}_{ij}(t)$ is the optimal permutation-invariant RMS distance
between the node-occupation trajectories of graphs $\mathcal G_i$ and
$\mathcal G_j$ up to time $t$. For both families, the worst-case
separation becomes and remains strictly positive, certifying that every pair of non-isomorphic graphs is distinguished by the one-point
occupation dynamics alone.}
\label{fig:srg_separability}
\end{figure*}

\section{Quantum Dynamics-based Attention Graph Transformer}
\label{sec:model}
The expressiveness of graph machine learning models is frequently benchmarked against the Weisfeiler-Leman (WL) hierarchy. Standard message-passing GNNs are bounded by the 1-WL test and often struggle with long-range dependencies~\cite{xu2018how,alon2021on}. While higher-order GNNs achieve greater expressiveness, they do so at the expense of steep polynomial scaling~\cite{morris2019wlgnn,maron2019provably}. Graph Transformers offer an alternative by capturing global interactions through self-attention, yet their expressiveness remains weak without carefully designed structural biases~\cite{dwivedi2020generalization,ying2021graphormer}. Consequently, advancing the state-of-the-art often requires injecting increasingly powerful (though computationally demanding) positional encodings.
The expressiveness mechanism considered in this work is not a fixed-level Weisfeiler--Leman refinement. For every fixed $k$, the $k$-WL test has families of non-isomorphic graphs that it cannot distinguish~\cite{cai1992optimal}. By contrast, the distinguishability criterion developed in section~\ref{subsec:problem_statement_theory} is stated at the level of graph-indexed Hamiltonian dynamics: for \emph{any graph}, distinct automorphism orbits induce, generically, distinct occupation-observable trajectories.

Applied to disjoint unions of connected graphs, this gives a generic separation mechanism for non-isomorphic graphs that is not constrained by a predetermined finite tuple dimension, as in fixed-$k$ WL. This comparison should nonetheless be read with care: the $k$-WL cost quoted above is that of an \emph{exact, deterministic} refinement, whereas our separation is \emph{probabilistic} and its practical cost is governed by the number of measurement shots, which scales with the inverse square of the relevant expectation-value gap. We therefore do not claim a uniform computational advantage over WL. Rather, the mechanism is not tied to a fixed tuple order, and the resulting trade-off between expressiveness and sampling cost is instance-dependent.

In this section, we introduce the Quantum Dynamics-based Attention Graph Transformer (\QDAGer). Rather than optimizing for scalability, \QDAGer is designed for small graphs and explicitly prioritizes structural expressiveness by injecting the dynamics directly into the attention mechanism. In this way, the model turns the generic separation mechanism of section~\ref{subsec:problem_statement_theory} into concrete node- and pair-level structural features, namely time-resolved occupations and connected two-point correlators. 
\newline Theorem~\ref{thm:expressiveness} (\textit{cf.} section~\ref{subsec:loss}) is the architecture-level formulation of this statement. Using the update rules defined below (\textit{cf.} section~\ref{subsec:model_archi}), the theorem claims that, for any finite dataset of graphs, and for Haar-generic initial states together with almost every choice of sampling times, there exists a \QDAGer parameter setting whose node- and pair-level embeddings separate all non-isomorphic graph pairs up to relabeling. Thus, Theorem~\ref{thm:expressiveness} translates the generic Hamiltonian separation mechanism into an expressiveness result for the actual model architecture. The practical caveat is that, in a concrete implementation, the observed separation depends on the chosen Hamiltonian parameters, initialization, sampling times, and measurement precision, which determine the trade-off between expressiveness, simulation cost, and statistical robustness.
\subsection{Model architecture}
\label{subsec:model_archi}
Our architecture is directly inspired by the GraphGPS design~\cite{rampasek2022GPS}, and we incorporate additional inductive biases for graph transformers following the adaptations proposed in~\cite{ma2023GraphInductiveBiases}. In particular, we retain the general GPS-style blueprint while tailoring its components to the transformer-based setting considered in~\cite{ma2023GraphInductiveBiases}.
\newline To compute the quantum-dynamical features, we emulate the Hamiltonian of Eq.~\eqref{eq:ham} using the open-source library from \footnote{See \url{https://github.com/pasqal-io/emulators} for implementation details.}~\cite{bidzhiev2025efficientemulationneutralatom}. While tensor-network approaches, including efficient architectures such as matrix product states, could in principle provide accurate and scalable approximations for suitable graph structures and entanglement regimes, we choose to compute the time evolution via state-vector (SV) emulation. This avoids truncation or bond-dimension effects and provides precise reference dynamics for benchmarking against state-of-the-art methods. This choice comes at a significant memory and runtime cost, limiting our simulations in practice to graphs with up to \(\sim 25\) nodes and motivating the use of small-graph datasets. Since the benchmark datasets considered here are not restricted to graphs admitting a direct neutral-atom hardware embedding, we rely on emulation of Eq.~\eqref{eq:ham} to ensure a controlled and comprehensive comparison across all datasets, as presented in section~\ref{sec:experiment}.
\newline A natural way to remain within these computational constraints while still training expressive transformer models is to consider tasks defined on \emph{pairs} of graphs. Indeed, emulating \(n\) individual graphs enables the construction of \(n(n-1)/2\) distinct unordered pairs, effectively increasing the number of supervised instances by a quadratic factor. This pairing strategy allows us to amortize the cost of expensive SV simulations and obtain datasets large enough to support higher-capacity models. Accordingly, we train our models on the datasets introduced in~\cite{JainNeuralGED}, where the limit for maximum graph size is $20$ nodes. 
Specifically, graphs with fewer than 20 nodes have their node and node-pair features padded with zeros up to the common maximum size. Finally, each dataset contains at most 1000 graphs. 
\newline This setup yields on the order of \(10^{5}\) graph pairs, which is sufficient for training small transformer architectures. The associated prediction task is graph edit distance (GED), a widely used measure of graph similarity based on the minimum-cost sequence of edit operations transforming one graph into another. Notably, exact GED computation is NP-hard~\cite{GEDComplexity}, making GED prediction a challenging benchmark for evaluating graph representation learning.
We now describe the model architecture used for this task, with Figure~\ref{fig:main} displaying the full pipeline described in the remainder of this section. 
\begin{figure*}[t]
	\centering
	\includegraphics[width=\textwidth]{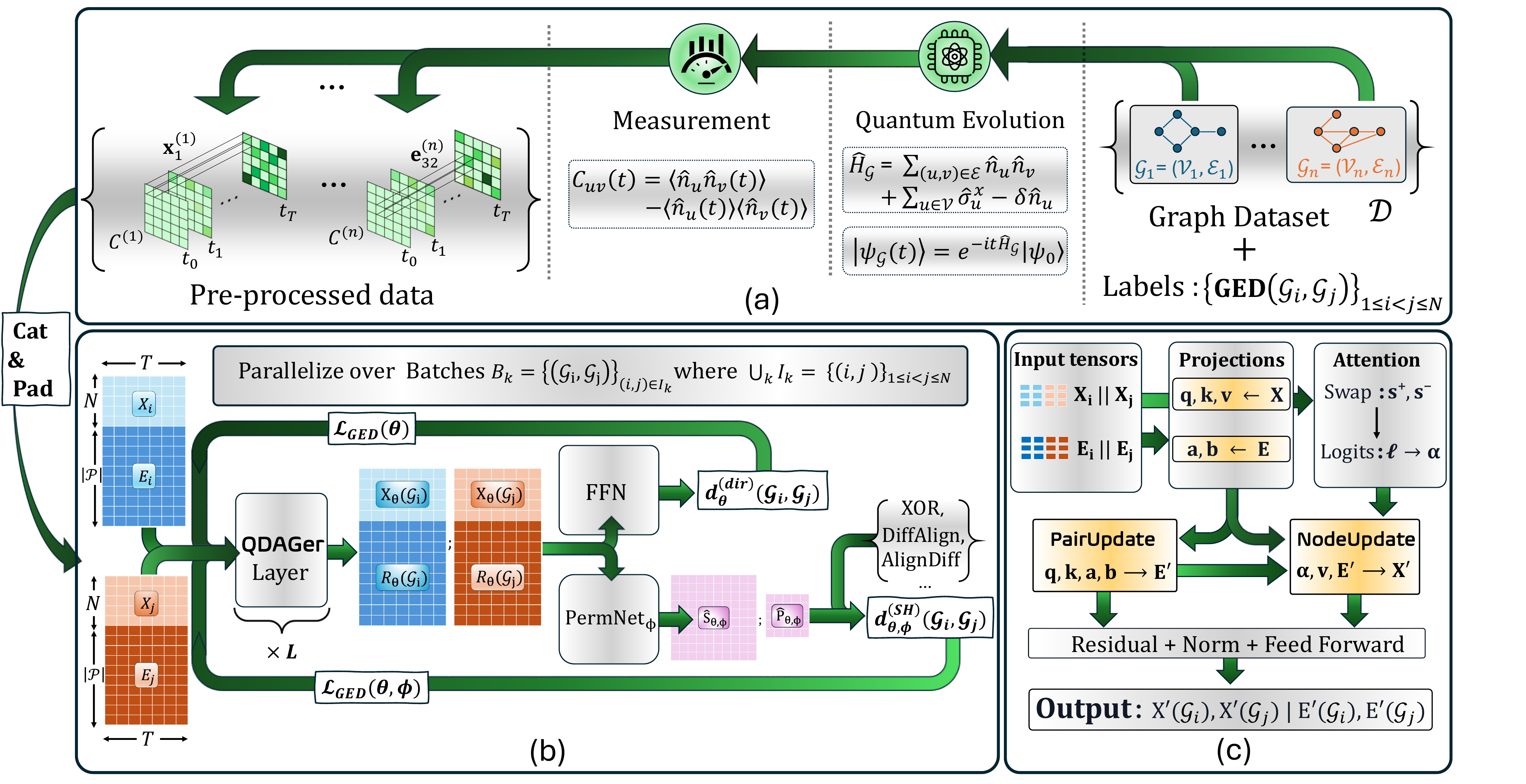}
\caption{
Detailed pipeline of the \QDAGer model architecture.
\textbf{(a)} Pre-processing phase: each graph $\mathcal{G}_i=(\mathcal{V}_i,\mathcal{E}_i)$ is mapped to the Hamiltonian of Eq.~\eqref{eq:ham}.
Starting from a common initial state $\ket{\psi_0}$, we evolve the system as
$\ket{\psi_{\mathcal{G}}(t)}=\exp(-it\hat{H}_{\mathcal{G}})\ket{\psi_0}$
and evaluate the relevant observables at multiple times, namely the node occupations
$\langle \hat{n}_u\rangle(t)$ and connected two-point correlators
$C_{uv}(t)=\langle \hat{n}_u\hat{n}_v\rangle(t)-\langle \hat{n}_u\rangle(t)\langle \hat{n}_v\rangle(t)$.
\textbf{(b)} For each graph pair $(\mathcal{G}_i,\mathcal{G}_j)$, the corresponding tensors are concatenated, padded, and processed by $\times L$  \QDAGer layers to produce node embeddings $X_\theta$ and pair embeddings $R_\theta$. These embeddings are used either by a direct distance head, yielding $d^{(\mathrm{dir})}_\theta(\mathcal{G}_i,\mathcal{G}_j)$, or by an alignment-based Sinkhorn branch in which \texttt{PermNet}$_\phi$ predicts soft node and pair alignments, yielding $d^{(\mathrm{SH})}_{\theta,\phi}(\mathcal{G}_i,\mathcal{G}_j)$. Both objectives aim to predict the corresponding ground-truth GED labels.
\textbf{(c)} One \QDAGer layer. Swap-consistent symmetric and antisymmetric features $(\mathbf{s}^{+},\mathbf{s}^{-})$ produce attention weights $\mathbf{\alpha}$, while linear projections generate node quantities $(\mathbf{q},\mathbf{k},\mathbf{v})$ and pair quantities $(\mathbf{a},\mathbf{b})$. Pair embeddings are updated through gated node-pair interactions, and node embeddings are updated using attention-weighted values together with the updated pair information, followed by residual, normalization, and feed-forward operations.
}
	\label{fig:main}
\end{figure*}

\subsubsection{Quantum dynamical features.}
For each graph $\mathcal{G}=(\mathcal{V},\mathcal{E})$ with $|\mathcal{V}|=N$ we emulate the state-vector dynamics and record
two-point correlators $C^{\mathcal{G}}(t)\in\mathbb{R}^{N\times N}$ at sampling (positive) times $t_1,\dots,t_T$. We associate to every node and every unordered pair the corresponding time-evolving observable : 
$\mathbf{x}_i\in\mathbb{R}^{T}:\ [\mathbf{x}_i]_t = \braket{\hat{n}_i(t)}_{\mathcal{G}}$ and $
\mathbf{e}_{ij}\in\mathbb{R}^{T}:\ [\mathbf{e}_{ij}]_t = \braket{\hat{n}_i\hat{n}_j(t) }_{\mathcal{G}}-  \braket{\hat{n}_i(t)}_{\mathcal{G}}\braket{\hat{n}_j(t)}_{\mathcal{G}}
\ (i<j).$
For a training pair $(\mathcal{G}_1,\mathcal{G}_2)$ we write $(\mathbf{x}^{(g)}_i,\mathbf{e}^{(g)}_{ij})$ for the features of $\mathcal{G}_g$, with $g\in\{1,2\}$.
We use $\langle \mathbf{u},\mathbf{v}\rangle := \mathbf{u}^\top\mathbf{v}=\sum_{t=1}^{T}u_t v_t$.

\subsubsection{Single-head quantum attention over nodes and pairs of nodes.}
Given $(\mathcal{G}_1,\mathcal{G}_2)$, for each node $i$ and pair $(i,j)$ we form
$
\mathbf{s}^{\pm}_i := \mathbf{x}^{(1)}_i \pm \mathbf{x}^{(2)}_i$ and 
$\mathbf{s}^{\pm}_{ij} := \mathbf{e}^{(1)}_{ij} \pm \mathbf{e}^{(2)}_{ij}.
$
With learnable $\mathbf{w}^{+},\mathbf{w}^{-}\in\mathbb{R}^{T}$ and biases $b \in\mathbb{R}$, define :
\begin{align}
\ell^{(1)}_{p}&=\langle \mathbf{s}^{+}_{p},\mathbf{w}^{+}\rangle+\langle \mathbf{s}^{-}_{p},\mathbf{w}^{-}\rangle+b, \\
\ell^{(2)}_{p}&=\langle \mathbf{s}^{+}_{p},\mathbf{w}^{+}\rangle-\langle \mathbf{s}^{-}_{p},\mathbf{w}^{-}\rangle+b,
\end{align}
for $p\in \mathcal{V}\cup\{(i,j):i<j\}$. 
Attention weights are then given by a softmax over the \emph{union} of node and pair indices:
\begin{multline}
\alpha^{(g)}_{p}= \text{softmax}(\ell^{(g)}_p)\;, \; p\in \mathcal{V}\cup\{(i,j):i<j\}, \\ g\in\{1,2\}.
\end{multline}
Here $\mathbf{s}^{+}_{p},\mathbf{s}^{-}_{p}$ are the symmetric/antisymmetric
combinations entering the logits. Under the swap $\mathcal{G}_1 \leftrightarrow \mathcal{G}_2$
the antisymmetric part changes sign, so the two channels are exchanged rather than left
invariant, $\ell^{(1)}_{p}\leftrightarrow\ell^{(2)}_{p}$ (and likewise
$\alpha^{(1)}_{p}\leftrightarrow\alpha^{(2)}_{p}$). The construction is thus
swap-\emph{equivariant}, and yields a swap-invariant output once the channels are recombined
symmetrically (\textit{cf.} section~\ref{subsec:loss}). This property is useful when the target GED is symmetric.
\subsubsection{Layer update.}
The update mechanism is carried out using two distinct mechanisms that we note \texttt{PairUpdate} and \texttt{NodeUpdate}
We map features to a $D$-dimensional latent space (all $W_\bullet\in\mathbb{R}^{D\times T}$):
\begin{align}
\mathbf{q}^{(g)}_i&=W_Q\mathbf{x}^{(g)}_i, \quad
\mathbf{k}^{(g)}_i=W_K\mathbf{x}^{(g)}_i, \quad
\mathbf{v}^{(g)}_i=W_V\mathbf{x}^{(g)}_i, \\
\mathbf{a}^{(g)}_{ij}&=W_{Ew}\mathbf{e}^{(g)}_{ij}, \quad
\mathbf{b}^{(g)}_{ij}=W_{Eb}\mathbf{e}^{(g)}_{ij}.
\end{align}
\texttt{PairUpdate}:
Pairs are updated by the following gated interaction:
\begin{multline}
\mathbf{e}^{\prime (g)}_{ij}
:=
\sigma\!\Big(\rho(\mathbf{q}^{(g)}_i+\mathbf{k}^{(g)}_j)\odot \mathbf{a}^{(g)}_{ij}+\mathbf{b}^{(g)}_{ij}\Big)
\in\mathbb{R}^{D} \\ \text{and we put : }
\rho(\mathbf{u})=\operatorname{sign}(\mathbf{u})\odot \log({1+|\mathbf{u}|})
\end{multline}
(all elementwise) to stabilize the training. Optionally, residuals may be used:
$\mathbf{e}^{\prime (g)}_{ij}\leftarrow \mathbf{e}^{\prime (g)}_{ij}+\mathbf{a}^{(g)}_{ij}$.
\newline
\texttt{NodeUpdate}:
Nodes mix an attention-weighted value term and an attention-weighted aggregation of updated pair features:
\begin{equation}
\label{eq:node_update}
\mathbf{x}^{\prime (g)}_{i}
:=
\alpha^{(g)}_{i}\,\mathbf{v}^{(g)}_{i}
+\sum_{j\neq i} \alpha^{(g)}_{ij}\,W_{NE}\mathbf{e}^{\prime (g)}_{i,j}
\in\mathbb{R}^{D}.
\end{equation}
\newline We then apply a standard feed-forward network (FFN) block with residual connection and normalization to $\{\mathbf{x}^{\prime (g)}_i\}$ and $\{\mathbf{e}^{\prime (g)}_{ij}\}$, as done in~\cite{ma2023GraphInductiveBiases}. Note that it is possible and straightforward to extend this model with multi-headed attention, each producing per-head updates ${\mathbf{x}_i^{\prime (g,h)}}$ and ${\mathbf{e}_{ij}^{\prime (g,h)}}$ that we later combine through linear layers via ${\mathbf{x}_i^{\prime (g)}} = \sum_{h=1}^H W_{X}.{\mathbf{x}_i^{\prime (g,h)}} + \mathbf{b}_{X}$ and ${\mathbf{e}_{ij}^{\prime (g)}} = \sum_{h=1}^H W_E.{\mathbf{e}_{ij}^{\prime (g,h)}} + \mathbf{b}_E$ where $W_{X/E}$, $\mathbf{b}_{X/E}$ are learned parameters. 

\subsection{Loss computation: direct embeddings vs.\ neural set divergence}
\label{subsec:loss}

Let a training sample be a pair of graphs $(\mathcal{G}_k,\mathcal{G}_l)$ with ground-truth edit distance $d(\mathcal{G}_k,\mathcal{G}_l)$.
\QDAGer layers (denoted by parameters $\theta$) produce node and pair embeddings
\begin{multline}
(X_\theta(\mathcal{G}_k),R_\theta(\mathcal{G}_k))\in\mathbb{R}^{N\times D}\times\mathbb{R}^{|\mathcal P|\times D} \;\; ; \;\;
\\
(X_\theta(\mathcal{G}_l),R_\theta(\mathcal{G}_l))\in\mathbb{R}^{N\times D}\times\mathbb{R}^{|\mathcal P|\times D}.
\end{multline}
With $N$ the number of nodes (including padding) and $\mathcal{P}$ the set of unordered node pairs with $|\mathcal{P}| = \frac{N(N-1)}{2}$
\subsubsection{Blindfold loss.}
We define a model distance by applying a permutation-invariant read-out $\Phi$ to the two embedding pairs
\begin{equation}
\label{eq:blindfold_loss}
d_\theta(\mathcal{G}_k,\mathcal{G}_l)
:=\Phi\!\Big(X_\theta(\mathcal{G}_k),R_\theta(\mathcal{G}_k),X_\theta(\mathcal{G}_l),R_\theta(\mathcal{G}_l)\Big),
\end{equation}
where $\Phi$ is any map invariant to permutations of the nodes and pairs of each graph. 
In our implementation, $\Phi$ consists of three steps: a pooling that reduces each embedding
set to a fixed-size descriptor, a symmetric/antisymmetric combination of the two graphs' descriptors, and a final FFN that outputs the scalar distance. We detail each below.

The pooling maps a graph's stacked embeddings $Z=(z_1,\dots,z_N)\in\mathbb{R}^{N\times D}$ (zero-padded to a common size $N$, with validity mask $\mathbf{M}\in\{0,1\}^{N}$ and valid-token count
$c=\sum_i M_i$) to a descriptor in $\mathbb{R}^{D}$,
\begin{equation}
\operatorname{pool}(Z)=\tfrac{1}{c}\textstyle\sum_{i} M_i\,z_i .
\end{equation}
The mean is thus mask-aware: it averages only over the $c$ unpadded tokens, so it is a genuine
average over the graph's real nodes (resp.\ pairs), independent of the padding length $N$.
Applying it to node and pair embeddings yields the descriptors
$\bar{X}_g=\operatorname{pool}(X_\theta(\mathcal{G}_g))$ and
$\bar{R}_g=\operatorname{pool}(R_\theta(\mathcal{G}_g))$, with $\bar{X}_g,\bar{R}_g\in\mathbb{R}^{D}$
and $g\in\{k,l\}$.

The two graphs' descriptors are then combined through symmetric ($+$) and antisymmetric ($-$) terms,
\begin{multline}
\label{eq:concat_rep}
\mathbf{h}_{kl}:=\operatorname{concat}\big(\bar{X}_k+\bar{X}_l,\ \bar{X}_k-\bar{X}_l,\
\\
\bar{R}_k+\bar{R}_l,\ \bar{R}_k-\bar{R}_l\big)\in\mathbb{R}^{4D},
\end{multline}
and a final FFN maps this to a scalar,
\begin{equation}
d_\theta(\mathcal{G}_k,\mathcal{G}_l):=\operatorname{FFN}_\theta(\mathbf{h}_{kl}),
\end{equation}
so that $\Phi=\operatorname{FFN}_\theta\circ\,\mathbf{h}\circ\operatorname{pool}$. Since
$\operatorname{pool}$ is permutation-invariant, so is $d_\theta$. The corresponding regression objective is
\begin{equation}
\mathcal{L}_{\mathrm{emb}}(\theta):=\mathbb{E}_{(\mathcal{G}_k,\mathcal{G}_l)}
\Big(d_\theta(\mathcal{G}_k,\mathcal{G}_l)-d(\mathcal{G}_k,\mathcal{G}_l)\Big)^2,
\end{equation}
with the empirical version obtained by averaging over mini-batches of graph pairs.
\subsubsection{Neural set divergence : alignment \& surrogate loss.}
In~\cite{JainNeuralGED}, authors introduce a method that learns the alignment between $\mathcal{G}_k$ and $\mathcal{G}_l$ through a doubly-stochastic differentiable permutation matrix produced by a neural network PermNet, parameterized by $\phi$, using the Sinkhorn (SH)~\cite{Mena2018LearningLP} module. They consequently introduce the matrices $\widehat P_{\theta,\phi}(\mathcal{G}_k,\mathcal{G}_l)\in\mathbb{R}^{N\times N}$ and $\widehat S_{\theta,\phi}(\mathcal{G}_k,\mathcal{G}_l)\in\mathbb{R}^{|\mathcal P|\times|\mathcal P|},$ where $\widehat P_{\theta,\phi}$ aligns nodes and $\widehat S_{\theta,\phi}$ (constructed either from $\widehat P_{\theta,\phi}$, from pair embeddings, or jointly) aligns pairs). Using these transport plans, they define an aligned surrogate distance
\begin{align}
\label{eq:surrogate_dist}
d_{\theta,\phi}(\mathcal{G}_k,\mathcal{G}_l)
:=\ &\mathfrak{S}\!\Big(
X_\theta(\mathcal{G}_k),\,R_\theta(\mathcal{G}_k), \notag\\
&\phantom{\mathfrak{S}\!\Big(}
\widehat P_{\theta,\phi}(\mathcal{G}_k,\mathcal{G}_l)\,X_\theta(\mathcal{G}_l), \notag\\
&\phantom{\mathfrak{S}\!\Big(}
\widehat S_{\theta,\phi}(\mathcal{G}_k,\mathcal{G}_l)\,R_\theta(\mathcal{G}_l)
\Big),
\end{align}
where $\mathfrak{S}$ is one of the surrogates (\texttt{AlignDiff} / \texttt{DiffAlign} / \texttt{XOR-DiffAlign}) introduced in section 4.2 of~\cite{JainNeuralGED}. The training objective is
$\mathcal L_{\mathrm{GED}}(\theta,\phi) :=\mathbb{E}_{(\mathcal{G}_k,\mathcal{G}_l)}\Big(d_{\theta,\phi}(\mathcal{G}_k,\mathcal{G}_l)-d(\mathcal{G}_k,\mathcal{G}_l)\Big)^2$, again with empirical loss computed by averaging over mini-batches of graph pairs. In practice, gradients flow through both the embedder (\QDAGer, via $X_\theta(\cdot),R_\theta(\cdot)$) and the alignment module (PermNet, via $\widehat P_{\theta,\phi},\widehat S_{\theta,\phi}$), enabling end-to-end training.

\subsubsection{Swap invariance.}
Having established the invariance of the \QDAGer layer under the swap of its two arguments, \(\mathcal{G}_k \leftrightarrow \mathcal{G}_l\), a remaining subtlety is whether this property persists at the level of the end-to-end model when these layers are considered jointly with the loss-specific layers described in the above paragraphs. This is desirable when the target distance itself has the property, \textit{i.e.} the ground truth satisfies $d(\mathcal{G}_k,\mathcal{G}_l)=d(\mathcal{G}_l,\mathcal{G}_k)$, which happens when the edit costs are symmetric, \textit{i.e.} equal node insertion/deletion costs, and equal edge insertion/deletion costs. With asymmetric (directed) edit costs the target is order-dependent, and swap invariance of either model is neither expected nor desirable. Since we do evaluate our models on a dataset where the edit costs are symmetric (\textit{cf.} section~\ref{sec:experiment}), it is worth examining whether each of the two variants of our model actually verify this property, and how it can be enforced when it does not.
\newline For the blindfold read-out of Eq.~\eqref{eq:blindfold_loss}, $d_\theta$ is invariant to node (and induced-pair) permutations, but it is \emph{not} invariant under the swap $\mathcal{G}_k\leftrightarrow\mathcal{G}_l$, and this holds irrespective of the edit costs: the antisymmetric terms $\bar{X}_k-\bar{X}_l$ and $\bar{R}_k-\bar{R}_l$ change sign under the exchange and are then mixed by a generic FFN. One could enforce swap invariance by replacing these terms by their respective absolute values in Eq.~\eqref{eq:concat_rep}, but we empirically observed that this compromises the performance of the corresponding model.
\newline For the surrogate distance of Eq.~\eqref{eq:surrogate_dist}, the situation is more intricate. Under the symmetric-cost assumption above, $d_{\theta,\phi}$ inherits the identity $d_{\theta,\phi}(\mathcal{G}_k,\mathcal{G}_l)=d_{\theta,\phi}(\mathcal{G}_l,\mathcal{G}_k)$ provided that: (i) a two-sided surrogate is used (this is verified by the \texttt{XOR-DiffAlign} function, which is our chosen surrogate in all the subsequent experiments), and (ii) the Sinkhorn iterations have converged to a genuinely doubly-stochastic plan, whereby $\widehat P_{\theta,\phi}(\mathcal{G}_l,\mathcal{G}_k)=\widehat P_{\theta,\phi}(\mathcal{G}_k,\mathcal{G}_l)^{\top}$ (and likewise for $\widehat S_{\theta,\phi}$). When any of these conditions is relaxed the estimator is only approximately symmetric. Exact symmetry can nonetheless be enforced \emph{a posteriori} through the symmetrized read-out $\tfrac{1}{2}\big(d_{\theta,\phi}(\mathcal{G}_k,\mathcal{G}_l)+d_{\theta,\phi}(\mathcal{G}_l,\mathcal{G}_k)\big)$.
\subsubsection{Model expressiveness.}
The losses above define how the embeddings produced by \QDAGer are used for
GED prediction. Independently of the particular regression head (FFN) or surrogate loss, one can ask whether the architecture has enough capacity to preserve the separating information contained in the dynamical features. The following result formalizes this expressiveness property under the generic separation assumptions of section~\ref{subsec:problem_statement_theory}.

Let \(\mathcal D\) be a finite dataset of graphs, and generate the dynamical input features of \QDAGer from initial states sampled according to the Haar measure on the corresponding trivial symmetry sectors.
\begin{theorem}[Expressiveness of \QDAGer]
\label{thm:expressiveness}
For Haar-almost every choice of initial states and for almost every choice of sampling times, there exists a parameter setting \(\theta_0\) such that, for every non-isomorphic pair \(\mathcal G_k\not\simeq \mathcal G_l\) in \(\mathcal D\), the resulting \QDAGer embeddings are distinct up to node relabeling. 
\end{theorem}
The almost-sure separation follows from Theorem~\ref{thm:main}, together with the fact that \(\mathcal D\) contains only finitely many graph pairs and that a finite union of zero-measure exceptional sets still has measure zero. In other words, under the Haar-random initialization considered in section~\ref{subsec:problem_statement_theory}, \QDAGer contains a parameter setting that preserves the separating dynamical features on any finite dataset, for almost every initial state and almost every choice of sampling times. The proof is provided in Appendix~\ref{app:thm_expressiveness}.
\newline Note that, in the following section, the experiments are performed with the deterministic initialization \(|0\rangle^{\otimes N}\). The theorem should therefore be read as a generic expressiveness guarantee, while the empirical section verifies that this physically natural initialization also yields informative features and a strong inductive bias in practice.

\section{Experimental Results}
\label{sec:experiment}
We evaluate on seven real-world graph datasets\footnote{Full code available in \url{https://github.com/pasqal-io/QDAGer}}: Mutagenicity (\textsc{Mutag}), \textsc{Ogbg-Code2} (\textsc{Code2}),
\textsc{Ogbg-Molhiv} (\textsc{Molhiv}), \textsc{Ogbg-Molpcba} (\textsc{Molpcba}), \textsc{AIDS}, \textsc{Linux}, and \textsc{Yeast}~\cite{JainNeuralGED}, which draw on the OGB suite~\cite{hu2020open}, the Mutagenicity and IAM graph databases~\cite{kazius2005derivation,riesen2008iam}, and the AIDS/Linux graph-similarity benchmarks~\cite{bai2019simgnn}. For consistency, we use the same train/validation/test splits and the same graph-pair construction protocol as in~\cite{JainNeuralGED}. The performances are reported via the mean-squared error, following the approaches described in section~\ref{subsec:loss}. Ground-truth GED values are computed exactly using the F2 solver implemented in GEDLIB~\cite{GEDLIB1, GEDLIB2}. Results are reported under two cost regimes: a uniform setting with $b^\ominus=b^\oplus=a^\ominus=a^\oplus=1$,
and a non-uniform setting with $b^\ominus=3,\ b^\oplus=1,\ a^\ominus=2,\ a^\oplus=1$, where $b^\ominus, b^\oplus, a^\ominus, a^\oplus$ respectively correspond to node removal, node addition, edge removal and edge addition costs~\cite{JainNeuralGED}.

We don't benchmark the labeled datasets, as the model described here is purely structure-oriented, and is not designed for labeled graphs.  We further follow~\cite{JainNeuralGED} for all baseline comparisons, using the same set of competing methods, namely GMN~\cite{li2019graph}, SimGNN~\cite{bai2019simgnn}, GraphSim~\cite{bai2020learning}, ISONET~\cite{roy2022interpretable}, GREED~\cite{ranjan2022greed}, ERIC~\cite{zhuo2022efficient}, H2MN~\cite{zhang2021h2mn}, EGSC~\cite{qin2021slow}, and GraphEDX~\cite{JainNeuralGED}. In addition, we run an ablation in which we replace the quantum-dynamics-derived inputs of \QDAGer by standard classical structural features, namely relative random-walk (RWs)~\cite{gartner2003graph,vishwanathan2010graph} and heat kernel (HK)~\cite{kondor2002diffusion} features computed on the input graphs. All other components (architecture, alignment module, and loss) are kept unchanged. Note that, for these experiments, the quantum dynamics are obtained by emulating expectation values, up to the emulation error discussed in~\cite{bidzhiev2025efficientemulationneutralatom}. The goal here is to assess the expressiveness of these features in an idealized setting, before eventually considering finite-shot estimation and hardware implementation on neutral-atom platforms in future work. All results are reported in Table~\ref{tab:res_table}, and additional experimental details are provided in Table~\ref{tab:qdager_equal_core} of Appendix~\ref{appx:expe_details_3}.  
\begin{table*}[t]
\centering
\scriptsize
\setlength{\tabcolsep}{3.6pt}
\renewcommand{\arraystretch}{1.18}
\caption{Comparison of graph edit distance prediction performance on both \textbf{uniform-cost} and \textbf{non-uniform-cost} settings (lower MSE is better). We report state of the art results together with our methods, grouped into \textbf{blindfold} (dir) and \textbf{surrogate} (SH) according to their corresponding model/loss. Each MSE provided with its standard error obtained on the test set. The best, second-best, and third-best results in each dataset column are highlighted in \bestcell{blue}, \secondcell{green}, and \thirdcell{orange}, respectively.}
\resizebox{\textwidth}{!}{%
\begin{tabular}{llccccccc}
\toprule
Cost & Method & AIDS & Yeast & Mutag & MolHIV & MolPCBA & Code2 & Linux \\
\midrule

\multirow{18}{*}{\rotatebox[origin=c]{90}{\textbf{Uniform cost}}}
& \multicolumn{8}{l}{\textbf{Blindfold loss}} \\
\cmidrule{2-9}
& GMN-Match
  & $0.821\pm0.010$ & $1.175\pm0.013$ & $0.797\pm0.013$ & $1.318\pm0.020$ & $1.073\pm0.011$ & $1.677\pm0.187$ & $0.687\pm0.088$ \\
& GMN-Embed
  & $1.044\pm0.013$ & $1.767\pm0.021$ & $1.032\pm0.016$ & $1.859\pm0.020$ & $1.951\pm0.020$ & $1.358\pm0.104$ & $0.736\pm0.102$ \\
& ISONET
  & $1.640\pm0.020$ & $1.578\pm0.019$ & $1.187\pm0.021$ & $1.354\pm0.015$ & $1.106\pm0.011$ & \thirdcell{$0.879\pm0.061$} & $1.185\pm0.115$ \\
& GREED
  & $1.004\pm0.012$ & $1.423\pm0.015$ & $1.398\pm0.033$ & $1.708\pm0.019$ & $1.550\pm0.017$ & $1.869\pm0.140$ & $1.331\pm0.169$ \\
& ERIC
  & $0.731\pm0.008$ & $0.969\pm0.010$ & $0.719\pm0.011$ & $1.165\pm0.018$ & $0.862\pm0.009$ & $1.363\pm0.110$ & $1.664\pm0.260$ \\
& SimGNN
  & $1.455\pm0.020$ & $1.999\pm0.043$ & $1.471\pm0.024$ & $1.609\pm0.020$ & $1.456\pm0.020$ & $2.667\pm0.215$ & $7.232\pm0.762$ \\
& H2MN
  & $1.114\pm0.015$ & $1.353\pm0.018$ & $1.278\pm0.021$ & $1.521\pm0.020$ & $1.402\pm0.020$ & $7.240\pm0.527$ & $2.238\pm0.247$ \\
& GraphSim
  & $1.936\pm0.026$ & $2.232\pm0.030$ & $2.005\pm0.031$ & $2.577\pm0.064$ & $1.656\pm0.023$ & $3.139\pm0.206$ & $2.900\pm0.318$ \\
& EGSC
  & $0.627\pm0.007$ & $0.950\pm0.010$ & $0.765\pm0.011$ & $1.138\pm0.016$ & $0.938\pm0.010$ & $4.165\pm0.285$ & $2.411\pm0.325$ \\
& \QDAGer (dir)
  & $0.630\pm0.008$ & $0.817\pm0.008$ & $0.794\pm0.011$ & $0.839\pm0.010$ & $0.752\pm0.007$ & $1.019\pm0.077$ & $1.056\pm0.124$ \\
& \QDAGer (HK+dir)
  & $0.597\pm0.007$ & $0.896\pm0.011$ & $0.811\pm0.012$ & $0.983\pm0.012$ & $0.811\pm0.008$ & $1.099\pm0.092$ & $1.464\pm0.182$ \\
& \QDAGer (RWs+dir)
  & $0.668\pm0.008$ & $0.967\pm0.011$ & $0.724\pm0.011$ & $0.827\pm0.009$ & $0.777\pm0.008$ & $1.007\pm0.097$ & $1.587\pm0.237$ \\
\cmidrule{2-9}
& \multicolumn{8}{l}{\textbf{Surrogate loss}} \\
\cmidrule{2-9}
& GRAPHEDX
  & $0.565\pm0.006$ & \secondcell{$0.717\pm0.007$} & \secondcell{$0.492\pm0.007$} & \thirdcell{$0.781\pm0.008$} & $0.764\pm0.007$ & \bestcell{$0.429\pm0.036$} & \bestcell{$0.354\pm0.043$} \\
& \QDAGer (SH)
  & \bestcell{$0.472\pm0.005$} & \bestcell{$0.699\pm0.007$} & \bestcell{$0.453\pm0.007$} & \bestcell{$0.578\pm0.006$} & \bestcell{$0.628\pm0.006$} & $1.035\pm0.081$ & \thirdcell{$0.589\pm0.093$} \\
& \QDAGer (HK+SH)
  & \secondcell{$0.486\pm0.005$} & \thirdcell{$0.717\pm0.008$} & $0.769\pm0.011$ & \secondcell{$0.687\pm0.007$} & \thirdcell{$0.671\pm0.006$} & $1.676\pm0.102$ & $4.078\pm1.479$ \\
& \QDAGer (RWs+SH)
  & \thirdcell{$0.495\pm0.005$} & $1.535\pm0.109$ & \thirdcell{$0.572\pm0.008$} & $0.853\pm0.011$ & \secondcell{$0.666\pm0.006$} & \secondcell{$0.555\pm0.045$} & \secondcell{$0.514\pm0.061$} \\

\midrule

\multirow{18}{*}{\rotatebox[origin=c]{90}{\textbf{Non-uniform cost}}}
& \multicolumn{8}{l}{\textbf{Blindfold loss}} \\
\cmidrule{2-9}
& GMN-Match
  & $31.522\pm0.513$ & $63.179\pm1.127$ & $69.210\pm0.883$ & $76.923\pm0.862$ & $23.985\pm0.224$ & $13.472\pm0.970$ & $21.519\pm2.256$ \\
& GMN-Embed
  & $33.221\pm0.523$ & $60.949\pm0.663$ & $72.495\pm0.915$ & $78.254\pm0.865$ & $28.437\pm0.268$ & $13.425\pm1.035$ & $20.591\pm2.136$ \\
& ISONET
  & $5.513\pm0.092$ & $4.555\pm0.061$ & $3.369\pm0.062$ & $3.451\pm0.039$ & $2.781\pm0.029$ & $3.025\pm0.206$ & $3.031\pm0.299$ \\
& GREED
  & $34.354\pm0.557$ & $60.652\pm0.704$ & $68.732\pm0.867$ & $78.300\pm0.795$ & $26.057\pm0.238$ & $11.095\pm0.773$ & $20.667\pm2.140$ \\
& ERIC
  & $1.581\pm0.017$ & $2.341\pm0.030$ & $1.981\pm0.032$ & $3.377\pm0.070$ & $2.057\pm0.020$ & $12.767\pm1.177$ & $7.809\pm0.911$ \\
& SimGNN
  & $4.316\pm0.071$ & $4.496\pm0.060$ & $4.747\pm0.079$ & $4.145\pm0.051$ & $3.465\pm0.047$ & $5.212\pm0.360$ & $5.369\pm0.546$ \\
& H2MN
  & $3.105\pm0.043$ & $3.678\pm0.046$ & $3.413\pm0.053$ & $3.782\pm0.046$ & $3.396\pm0.046$ & $9.435\pm0.728$ & $5.848\pm0.611$ \\
& GraphSim
  & $5.266\pm0.081$ & $6.907\pm0.137$ & $5.370\pm0.092$ & $6.643\pm0.181$ & $3.928\pm0.053$ & $7.405\pm0.577$ & $6.815\pm0.628$ \\
& EGSC
  & $1.693\pm0.023$ & $2.157\pm0.027$ & $1.758\pm0.026$ & $2.371\pm0.025$ & $2.133\pm0.022$ & $3.957\pm0.365$ & $5.503\pm0.496$ \\
& \QDAGer (dir)
  & $1.382\pm0.018$ & $1.867\pm0.019$ & $2.038\pm0.028$ & $1.779\pm0.020$ & $1.717\pm0.017$ & $2.387\pm0.203$ & \secondcell{$2.050\pm0.228$} \\
& \QDAGer (HK+dir)
  & $1.421\pm0.017$ & $2.053\pm0.023$ & $1.949\pm0.028$ & $2.069\pm0.024$ & $1.910\pm0.019$ & $2.427\pm0.203$ & $3.031\pm0.402$ \\
& \QDAGer (RWs+dir)
  & $1.422\pm0.016$ & $2.104\pm0.023$ & $1.902\pm0.029$ & $2.052\pm0.022$ & $1.795\pm0.018$ & \thirdcell{$2.012\pm0.151$} & \thirdcell{$2.909\pm0.335$} \\
\cmidrule{2-9}
& \multicolumn{8}{l}{\textbf{Surrogate loss}} \\
\cmidrule{2-9}
& GRAPHEDX
  & $1.252\pm0.014$ & \secondcell{$1.603\pm0.016$} & \bestcell{$1.134\pm0.016$} & $1.804\pm0.019$ & $1.677\pm0.016$ & \bestcell{$1.478\pm0.118$} & \bestcell{$0.914\pm0.110$} \\
& \QDAGer (SH)
  & \bestcell{$1.142\pm0.012$} & \bestcell{$1.541\pm0.015$} & \secondcell{$1.194\pm0.017$} & \bestcell{$1.432\pm0.014$} & \bestcell{$1.439\pm0.013$} & $2.368\pm0.168$ & $12.469\pm4.291$ \\
& \QDAGer (HK+SH)
  & \thirdcell{$1.181\pm0.012$} & \thirdcell{$1.722\pm0.020$} & \thirdcell{$1.275\pm0.018$} & \secondcell{$1.508\pm0.016$} & \secondcell{$1.523\pm0.014$} & $3.037\pm0.216$ & $33.715\pm8.329$ \\
& \QDAGer (RWs+SH)
  & \secondcell{$1.154\pm0.013$} & $1.769\pm0.019$ & $1.311\pm0.018$ & \thirdcell{$1.722\pm0.019$} & \thirdcell{$1.585\pm0.015$} & \secondcell{$1.758\pm0.126$} & $85.321\pm31.820$ \\
\bottomrule
\end{tabular}%
}
\label{tab:res_table}
\end{table*}
\subsection{Discussion}
\label{sec:discussion}

This work proposes \QDAGer, a graph pair transformer, trained for GED prediction, that injects features derived from Ising Hamiltonian dynamics into the attention mechanism. The key contribution is not simply architectural, but the introduction of a tunable, dynamical structural signal (time-series of node occupations and connected two-point correlators) motivated by an analysis of how graph-indexed Hamiltonian dynamics interact with graph symmetries.

\subsubsection{Performance in the data-rich regime.}
A consistent trend in our results is that \QDAGer is most effective when the training set is sufficiently large to fit higher-capacity, structure-aware models. In particular, on datasets with on the order of $10^5$ graph pairs or more (YEAST, MUTAGENICITY, AIDS, OGBG-MOLHIV, and OGBG-MOLPCBA), \QDAGer predominantly achieves the best performance and improves over the strongest baselines reported under the same evaluation protocol (Table~1). In contrast, on smaller datasets such as LINUX ($1{,}774$ pairs) and OGBG-CODE2 ($3{,}629$ pairs), results are naturally more sensitive to optimization variance and hyperparameter choices. This behavior is expected for transformer-like models, and may be amplified in settings that combine expressive embedding backbones with permutation-invariant discrepancies and/or alignment modules. Overall, these observations suggest that the proposed dynamical inductive bias is most beneficial when generalization is not dominated by data scarcity.

\subsubsection{Ablation study.}
A central question is whether the gains stem from the proposed quantum-dynamical features or merely from the model capacity and training pipeline. To isolate the effect of the input signal, we perform an ablation in which the quantum dynamics used by \QDAGer are replaced with classical relative random-walk (RWs), as well as heat kernel (HK) features (more details are provided in Appendix~\ref{appx:expe_details_2}) and,  while keeping the architecture, alignment module, and loss computation unchanged. On the larger datasets, where training is more stable, \QDAGer with quantum-dynamics-based attention outperforms its HK and RWs counterparts (Table~1) in most cases for the blindfold loss, and in all cases under the surrogate-loss settings.
This provides empirical evidence that the quantum-derived observables yield a stronger structural bias for GED prediction than the considered classical alternative. From the graph-ML perspective, this result is appealing because it suggests that the dynamical features act as more than a generic positional encoding: they form a task-relevant representation that is sensitive to symmetries and substructure. Intuitively, time-evolving local observables and correlators can be viewed as a multi-scale probe of the input graph, where the non-linear dependence on the Hamiltonian parameters and the evolution time enriches the representational family beyond static encodings.

\subsubsection{Surrogate loss protocol and comparability.}
We emphasize that all surrogate-loss experiments are reported under a conservative protocol: we do not perform surrogate selection or specific tuning, using only a fixed surrogate objective (\texttt{XOR-DiffAlign})~\cite{JainNeuralGED}. This contrasts with settings where selecting multiple surrogate variants and hyperparameters via additional model selection loops can materially affect performance. Our choice keeps comparisons controlled, attributing improvements primarily to the proposed dynamical signal and architecture rather than extra degrees of freedom at the loss level. Consequently, the results in Table~1 reflect performance under a fixed, standard surrogate configuration.

\subsubsection{Practical constraints and limitations.}
In this work, quantum-dynamical features are computed via state-vector emulation, which provides access to expectation values (up to emulator error) but limits feasible graph sizes and makes feature generation expensive. Although tensor-network methods, and in particular matrix product state approaches, can provide better scaling for suitable graph structures and entanglement regimes, our choice was intentional: it allows us to assess the representational benefit of the proposed signal under idealized conditions before introducing additional confounders, such as finite-shot noise or hardware imperfections, in future work.

Moving toward hardware-relevant deployments raises two immediate questions: how performance degrades under finite-shot estimation of observables, and how feature generation can be scaled to much larger graphs. These considerations are orthogonal to the learning architecture itself, but they are essential for translating dynamical attention mechanisms into practical at-scale graph-ML pipelines.

The hardware-compatible form described by Eq.~\eqref{eq:phys_motivation} also comes with specific constraints. A direct neutral-atom implementation would require a geometric embedding of the input graph into an array geometry, and not every abstract graph admits such an embedding under the desired interaction pattern. In particular, unit-disk constraints and residual long-range van der Waals interactions become relevant when one asks whether Eq.~\eqref{eq:ham} can be implemented natively on a given device. These constraints do not affect the abstract graph-indexed Hamiltonian studied in this paper, but they are important for future hardware realizations.

On the cost side, feature generation is a one-time preprocessing step whose expense is amortized across the \(O(n^2)\) graph pairs constructed from \(n\) emulated graphs. Scaling beyond the state-vector regime will most likely require more efficient methods, other approximate emulators, hardware execution for compatible instances, or learned surrogates that regress the dynamical features directly.

\section{Conclusion}
In this work, we investigated graph-indexed Ising Hamiltonian dynamics as a symmetry-consistent source of structural signal for graph comparison. The main takeaway is that time-dependent local measurements, such as node occupations and two-body correlators, can serve as informative, tunable probes of graph structure. This perspective motivates using quantum dynamics not merely as an exotic feature generator, but as a principled mechanism to build representations that respond to structural roles and symmetries rather than to arbitrary node labels.

Building on this idea, we introduced \QDAGer, a graph-pair Transformer that injects these dynamical features directly into the attention mechanism through swap-consistent symmetric/antisymmetric constructions, and jointly updates node and pair representations. We evaluated \QDAGer on graph edit distance prediction under both direct permutation-invariant discrepancies and alignment-based surrogate losses. Across several GED benchmarks, we observed that the proposed quantum-dynamical signal provides a strong inductive bias, and ablations replacing it with classical random-walk or heat-kernel features generally reduce performance under the same training protocol.

From a practical standpoint, the current results rely on state-vector emulation to generate dynamical features, which constrains feasible graph sizes and makes preprocessing expensive. Important directions for future work include relaxing these constraints by using more scalable simulation methods, such as tensor networks, or by evaluating compatible instances on neutral-atom quantum processing units~\cite{Henriet2020quantumcomputing, dalyac2024graphalgorithms}. This will require understanding the model's robustness under finite-shot estimation noise, as well as clarifying which graph families benefit most from this dynamical bias in hardware-relevant regimes.. More broadly, our results suggest that incorporating physically motivated dynamics into attention can be a viable route to more structure-aware graph Transformers for graph comparison, and more generally to graph learning tasks.

\clearpage
\bibstyle{sn-mathphys}
\bibliography{biblio}
\appendix
\section{Extended Formulation}

\label{app:basis_details}
\subsection{Graph theory}
\label{app:basis_details_1}
A graph is denoted $\mathcal{G}=(\mathcal{V},\mathcal{E})$, where $\mathcal{V}$ is the vertex set with
$|\mathcal{V}|=N$, and $\mathcal{E}\subseteq \mathcal{V}\times \mathcal{V}$ is the (undirected) edge set.
We write $(u,v)\in\mathcal{E}$ to denote an edge between $u$ and $v$.

\subsubsection{Graph relabeling.}
Fix an ordering of $\mathcal{V}$ so that any permutation $g\in S_N$ can be identified with a bijection
$g:\mathcal{V}\to\mathcal{V}$. We denote the action of $g$ on vertices by $u\mapsto g\star u$.
This action extends to subsets $S\subseteq \mathcal{V}$ by
\begin{equation}
g\star S := \{g\star u:\; u\in S\},    
\end{equation}
and therefore defines an action of $S_N$ on the power set $2^{\mathcal{V}}$.

The relabeled graph $g\star \mathcal{G}$ is defined as the graph with vertex set $\mathcal{V}$ and edge set
\begin{equation}
g\star \mathcal{E} \;:=\; \{(g\star u,\; g\star v):\; (u,v)\in\mathcal{E}\}.
\end{equation}
\subsubsection{Graph isomorphism.}
Two graphs $\mathcal{G}_1=(\mathcal{V}_1,\mathcal{E}_1)$ and $\mathcal{G}_2=(\mathcal{V}_2,\mathcal{E}_2)$ with
$|\mathcal{V}_1|=|\mathcal{V}_2|=N$ are \emph{isomorphic}, denoted $\mathcal{G}_1\simeq \mathcal{G}_2$, if there exists
a bijection $\varphi:\mathcal{V}_1\to \mathcal{V}_2$ such that for all $u,v\in\mathcal{V}_1$,
\begin{equation}
(u,v)\in\mathcal{E}_1\;\; \Longleftrightarrow \;\; (\varphi(u),\varphi(v))\in\mathcal{E}_2.
\end{equation}

\subsubsection{Graph automorphism.}
An \emph{automorphism} of $\mathcal{G}$ is a permutation $g\in S_N$ such that $g\star \mathcal{G}=\mathcal{G}$.
The set of all automorphisms forms a group under composition, denoted $\Gamma\subseteq S_N$.

\subsubsection{Orbits on vertices and subsets.}
We consider the induced action of $\Gamma$ on $2^{\mathcal{V}}$ described above. For any subset
$S\subseteq \mathcal{V}$, we define its \emph{orbit} under $\Gamma$ as
\begin{equation}
O(S) \;:=\; \{g\star S:\; g\in\Gamma\}\;\subseteq\; 2^{\mathcal{V}}.
\end{equation}
In particular, for a vertex $u\in\mathcal{V}$, its orbit is $O(\{u\})$.

Moreover, for any $S\subseteq \mathcal{V}$, the induced subgraphs $\mathcal{G}[S]$ and $\mathcal{G}[g\star S]$
are isomorphic via 
\begin{equation}
g|_S:S\to g\star S.
\end{equation}
 Hence the orbit $O(S)$ collects subsets that define isomorphic induced subgraphs of $\mathcal{G}$.

\subsubsection{Connection to bitstrings.}
We identify subsets of vertices with computational basis states via the bijection
\begin{equation}
\iota:\{0,1\}^{N}\to 2^{\mathcal{V}},\qquad 
\iota(\boldsymbol{b})\;:=\;\{u\in\mathcal{V}:\; b_u=1\},
\end{equation}
and conversely $\boldsymbol{b}=\mathbf{1}_S$ for a subset $S\subseteq\mathcal{V}$, where $\mathbf{1}_S\in\{0,1\}^N$ denotes the indicator (characteristic) bitstring of $S$, i.e. $(\mathbf{1}_S)_u=1$ if and only if $u\in S$. The action of $\Gamma$ on $\mathcal{V}$ induces an action on bitstrings by permutation of
coordinates:
\begin{equation}
(g\star \boldsymbol{b})_u \;:=\; b_{g^{-1}\star u},\qquad g\in\Gamma,
\end{equation}
so that the mapping is equivariant:
\begin{equation}
\iota(g\star \boldsymbol{b})= g\star \iota(\boldsymbol{b}).
\end{equation}
Therefore, the orbit of a subset $S$ under $\Gamma$ coincides (via $\iota^{-1}$) with the orbit of its
indicator bitstring $\mathbf{1}_S$:
\begin{equation}
O(S)=\{\iota(g\star \mathbf{1}_S):\; g\in\Gamma\}.
\end{equation}
In particular, all bitstrings in the same orbit are mapped to induced subgraphs that are isomorphic
via the restriction of an automorphism of $\mathcal{G}$.
\subsubsection{Direct product of groups.}
Given two groups $\Gamma_1$ and $\Gamma_2$, their \emph{direct product} is
\begin{equation}
\Gamma_1\times \Gamma_2 := \{(g_1,g_2):\; g_1\in\Gamma_1,\ g_2\in\Gamma_2\},
\end{equation}
with componentwise composition $(g_1,g_2)\circ (h_1,h_2) := (g_1\circ h_1,\ g_2\circ h_2)$.

\subsection{Coefficients of symmetric bitstrings}
\label{app:basis_details_2}
In the standard computational basis $\mathcal{B :=}\{\ket{\boldsymbol{b}}_i\}_{i = 1}^{2^N}$, where each bitstring is a binary chain of length $N$. When the initial state $\ket{\psi_0}$ belongs to the trivial symmetry sector $\mathcal H_0$ (\textit{cf.} section~\ref{subsec:graph_automorphisms_wavefunction_symmetries}), we can expand $\ket{\psi_{\mathcal{G}}(t)}$ as:
\begin{equation}
\ket{\psi_{\mathcal{G}}(t)} = \sum_{\boldsymbol{b} \in \{0,1\}^N} C_{\mathcal{G}}^{\boldsymbol{b}}(t) \ket{\boldsymbol{b}}
\end{equation}
If $\ket{\boldsymbol{b}}$ and $\ket{\boldsymbol{b}'}$ are in the same orbit (i.e., $\ket{\boldsymbol{b}'} = \hat{U}_g \ket{\boldsymbol{b}}$ for some $g \in \Gamma$), then:
\begin{equation}
\begin{aligned}
C_{\mathcal{G}}^{\boldsymbol{b}}(t) 
& \defeq \bra{\boldsymbol{b}} e^{-it\hat{H}_{\mathcal{G}}} \ket{\psi_0} 
= \bra{\boldsymbol{b}} e^{-it\hat{H}_{\mathcal{G}}} \hat{U}_g^\dagger \ket{\psi_0} \\
&\qquad \text{(since } \hat{U}_g \ket{\psi_0} = \ket{\psi_0} \; \forall g \in \Gamma)\\
&= \bra{\boldsymbol{b}} \hat{U}_g^\dagger e^{-it\hat{H}_{\mathcal{G}}} \ket{\psi_0} 
= \bra{g \star \boldsymbol{b}} e^{-it\hat{H}_{\mathcal{G}}} \ket{\psi_0} \\
&\qquad \text{(since } \big[\hat{U}_g,\hat{H}_{\mathcal G} \big] = 0)\\
&= C_{\mathcal{G}}^{\boldsymbol{b}'}(t)
\end{aligned}
\end{equation}
This confirms that bitstrings belonging to the same orbit of $\Gamma$ share the same time-evolving amplitude. 

\section{Proof of Theorem~\ref{thm:main}}
\label{app:proof_main}
\begin{proof}
Let \(\widehat{P}_0\) be the orthogonal projector onto the trivial symmetry sector \(\mathcal H_0\), and define the projected observable difference
\begin{equation}
	\widehat\Delta_0
	:=
	\widehat{P}_0\big(\widehat O_S-\widehat O_{S'}\big)\widehat{P}_0.
\end{equation}
Since \(\widehat{P}_0|\psi_0\rangle=|\psi_0\rangle\) and \(\mathcal H_0\) is invariant under \(\widehat H_\mathcal{G}\), the signal depends only on this compressed observable. Equivalently, in the trivial symmetry sector with the corresponding Hamiltonian block $\widehat H_{\mathcal G}|_{\mathcal H_0}$ denoted $\widehat{H}_{0,\mathcal{G}}$,
\begin{equation}
\label{eq:time_dep_sig}
	f_{S,S'}(t)
	=
	\big\langle \psi_0\big|
	e^{it\widehat H_{0,\mathcal{G}}}
	\widehat\Delta_0
	e^{-it\widehat H_{0,\mathcal{G}}}
	\big|\psi_0\big\rangle.
\end{equation}

We first establish that \(\widehat\Delta_0\) is not the zero operator.

\begin{lemma}[Distinct projected observables]
\label{lem:distinct_projected_observables}
If \(S,S'\subseteq \mathcal{V}\) belong to two distinct orbits in \(2^\mathcal{V}\), then
\begin{equation}
	\widehat{P}_0\big(\widehat O_S-\widehat O_{S'}\big)\widehat{P}_0 \neq 0.
\end{equation}
\end{lemma}

\begin{proof}[Proof of Lemma~\ref{lem:distinct_projected_observables}]
For \(X\subseteq \mathcal{V}\), denote by \(\mathcal O_X\) its orbit under the induced action of \(\Gamma\) on \(2^\mathcal{V}\), and define the normalized orbit state
\begin{equation}
	|\Psi_X\rangle
	:=
	\frac{1}{\sqrt{|\mathcal O_X|}}
	\sum_{Y\in \mathcal O_X}|Y\rangle,
\end{equation}
where \(|Y\rangle\) is the computational basis vector associated with the indicator vector of \(Y\). By construction, \(|\Psi_X\rangle\in\mathcal H_0\), and therefore \(\widehat{P}_0|\Psi_X\rangle=|\Psi_X\rangle\).

For every \(Y\subseteq \mathcal{V}\), the observable \(\widehat O_S=\prod_{u\in S}\widehat n_u\) satisfies
$
	\langle Y|\widehat O_S|Y\rangle
	=
	\mathbf 1_{\{S\subseteq Y\}}.
$
Hence
\begin{align}
	\langle \Psi_X|\widehat O_S|\Psi_X\rangle
	&=
	\frac{1}{|\mathcal O_X|}
	\sum_{Y\in\mathcal O_X}
	\langle Y|\widehat O_S|Y\rangle \\
	&=
	\frac{|\{Y\in\mathcal O_X : S\subseteq Y\}|}{|\mathcal O_X|}.
\end{align}
Thus, $\langle \Psi_X|\widehat O_S|\Psi_X\rangle$ is the fraction of configurations in the orbit $\mathcal O_X$ that contain $S$. Taking \(X=S\), all elements of \(\mathcal O_S\) have cardinality \(|S|\). Thus \(S\subseteq Y\), with \(Y\in\mathcal O_S\), holds if and only if \(Y=S\). Consequently,
$
	\langle \Psi_S|\widehat O_S|\Psi_S\rangle
	=
	\frac{1}{|\mathcal O_S|}.
$

We now examine the two cross-expectations. If
$
	\langle \Psi_S|\widehat O_{S'}|\Psi_S\rangle>0,
$
then there exists \(Y\in\mathcal O_S\) such that \(S'\subseteq Y\), and therefore \(|S'|\leq |S|\). Similarly, if
$
	\langle \Psi_{S'}|\widehat O_S|\Psi_{S'}\rangle>0,
$
then there exists \(Y'\in\mathcal O_{S'}\) such that \(S\subseteq Y'\), and therefore \(|S|\leq |S'|\). Hence both cross-expectations can be strictly positive only if \(|S|=|S'|\). In that case, the inclusion \(S'\subseteq Y\in\mathcal O_S\) forces \(S'=Y\), so \(S'\in\mathcal O_S\), contradicting the assumption that \(S\) and \(S'\) belong to distinct orbits.

Therefore, at least one of the two cross-expectations is zero. If
$
	\langle \Psi_S|\widehat O_{S'}|\Psi_S\rangle=0,
$
then
$
	\langle \Psi_S|\widehat\Delta_0|\Psi_S\rangle
	=
	\frac{1}{|\mathcal O_S|}
	\neq 0.
$
Otherwise,
$
	\langle \Psi_{S'}|\widehat O_S|\Psi_{S'}\rangle=0,
$
and then
$
	\langle \Psi_{S'}|\widehat\Delta_0|\Psi_{S'}\rangle
	=
	-\frac{1}{|\mathcal O_{S'}|}
	\neq 0.
$
In either case, \(\widehat\Delta_0\neq 0\).
\end{proof}

We now use the spectral decomposition of \(\widehat H_{0,\mathcal{G}}\). Write
$
	\widehat H_{0,\mathcal{G}}
	=
	\sum_{E\in\Sigma(\widehat H_{0,\mathcal{G}})} E \widehat{P}_E,
$
where $\Sigma(\widehat H_{0,\mathcal{G}})$ is the set of eigenvalues of $\widehat H_{0,\mathcal{G}}$ and \(\widehat{P}_E\) is the orthogonal projector onto the eigenspace associated with each of these eigenvalues \(E\). Let
\begin{equation}
	\mathcal B
	:=
	\{E-E' : E,E'\in\Sigma(\widehat H_{0,\mathcal{G}})\}
\end{equation}
be the finite set of energy gaps. For each energy gap $\omega$, we define the $\omega$-resolved transition component of the perturbation $\widehat\Delta_0$ with respect to the spectrum of $\widehat H_{0,\mathcal{G}}$ as
\begin{equation}
\widehat\Delta_\omega
:=
\sum_{\substack{E,E'\in\Sigma(\widehat H_{0,\mathcal{G}})\\ E-E'=\omega}}
\widehat{P}_E\widehat\Delta_0\widehat{P}_{E'}.
\end{equation}
This operator extracts the part of $\widehat\Delta_0$ that couples eigenspaces of $\widehat H_{0,\mathcal{G}}$ whose energies differ by $\omega$, i.e., transitions from energy $E'$ to energy $E=E'+\omega$.
By completeness of the spectral projectors, we have
$
	\widehat\Delta_0
	=
	\sum_{\omega\in\mathcal B}\widehat\Delta_\omega.
$
Since \(\widehat\Delta_0\neq 0\), there exists \(\omega_\star\in\mathcal B\) such that
$
	\widehat\Delta_{\omega_\star}\neq 0.
$

The time-dependent signal of Eq.~\eqref{eq:time_dep_sig} can now be expanded as
\begin{equation}
\begin{aligned}
	f_{S,S'}(t)
	&=
	\sum_{E,E'\in\Sigma(\widehat H_{0,\mathcal{G}})}
	e^{it(E-E')}
	\\
	&\qquad\times
	\big\langle \psi_0\big|
	\widehat{P}_E\widehat\Delta_0\widehat{P}_{E'}
	\big|\psi_0\big\rangle \\
	&=
	\sum_{\omega\in\mathcal B}
	e^{it\omega}
	\big\langle \psi_0\big|
	\widehat\Delta_\omega
	\big|\psi_0\big\rangle.
\end{aligned}
\end{equation}
Thus \(f_{S,S'}\) is a finite trigonometric polynomial, and hence a real-analytic function of \(t\). Since the energy gaps in \(\mathcal B\) are distinct, the functions
$
	t\mapsto e^{it\omega},
	\;\; \omega\in\mathcal B,
$
are linearly independent over \(\mathbb C\). Therefore, if \(f_{S,S'}\equiv 0\), then all its energy gaps coefficients vanish:
$
	\big\langle \psi_0\big|
	\widehat\Delta_\omega
	\big|\psi_0\big\rangle
	=0,
	\;\;
	\forall \omega\in\mathcal B.
$
In particular, we get
$
	\big\langle \psi_0\big|
	\widehat\Delta_{\omega_\star}
	\big|\psi_0\big\rangle
	=0.
$

We now show that this last condition holds only on a Haar-null subset of \(\mathbb S(\mathcal H_0)\). Let \(M=\dim(\mathcal H_0)\), fix an orthonormal basis of \(\mathcal H_0\), and identify \(\psi_0\) with a vector \(z\in\mathbb C^M\). \newline Define :
$
	Q(z)
	:=
	z^\dagger \widehat\Delta_{\omega_\star}z.
$
This is a complex-valued homogeneous polynomial of degree two in the \(2M\) real coordinates of \(z\).

We claim that \(Q\) is not identically zero. To see this, consider the sesquilinear form
\begin{equation}
	B(v,w)
	:=
	v^\dagger \widehat\Delta_{\omega_\star}w.
\end{equation}
For complex vector spaces, the polarization identity gives
\begin{multline}
	B(v,w)
	=
	\frac14
	\Big(
		Q(v+w)-Q(v-w)
	\\
		+iQ(v-iw)-iQ(v+iw)
	\Big).
\end{multline}
If \(Q(z)=0\) for every \(z\in\mathbb C^M\), then the right-hand side vanishes for all \(v,w\in\mathbb C^M\). Hence
$
	v^\dagger \widehat\Delta_{\omega_\star}w=0
	\;\;
	\forall v,w\in\mathbb C^M,
$
which implies \(\widehat\Delta_{\omega_\star}=0\), contradicting the choice of \(\omega_\star\). Therefore \(Q\) is not identically zero. Equivalently, at least one of the real polynomials \(\operatorname{Re}Q\) and \(\operatorname{Im}Q\) is not identically zero.

It follows that the set
$
	Z_{\omega_\star}
	:=
	\{z\in\mathbb C^M : Q(z)=0\}
$
is a proper algebraic zero set (that is, the zero set of a polynomial which is not identically zero, and therefore not the whole ambient space). Since \(Q\) is homogeneous, if its restriction to the unit sphere were identically zero, then \(Q\) would be identically zero on all of \(\mathbb C^M\), which we have just ruled out. Thus the restriction of at least one of \(\operatorname{Re}Q\) or \(\operatorname{Im}Q\) to \(\mathbb S(\mathcal H_0)\) is a nonzero real-analytic function. Its zero set has Haar measure zero. Consequently,
\begin{equation}
	\mu_{\mathrm{Haar}}
	\bigl(
		Z_{\omega_\star}\cap \mathbb S(\mathcal H_0)
	\bigr)
	=0.
\end{equation}
Hence, for \(\mu_{\mathrm{Haar}}\)-almost every \(\psi_0\in\mathbb S(\mathcal H_0)\),
$
	\big\langle \psi_0\big|
	\widehat\Delta_{\omega_\star}
	\big|\psi_0\big\rangle
	\neq 0,
$
and therefore \(f_{S,S'}\not\equiv 0\).

Finally, a nonzero real-analytic function on \(\mathbb R\) has isolated zeros. In particular,
\begin{equation}
	\lambda_1\bigl(\{t\in\mathbb R : f_{S,S'}(t)=0\}\bigr)=0.
\end{equation}
This completes the proof.
\end{proof}
\begin{remark}[Exceptional initial states for the Transverse-Field Ising Hamiltonian]
\label{rem:pathological_initial_states}

As pointed out in section~\ref{subsec:empirical_validation}, for a deterministic initialization, we can no longer assume Haar-randomness. Therefore, the relevant issue concerns which part of $\mathcal H_0$ is actually explored by the dynamics. Given a \emph{fixed} $|\psi_0\rangle\in\mathcal H_0$, define its dynamically accessible subspace
\begin{equation}
\label{eq:def_accessible_subspace}
\mathcal K(\psi_0)
:=
\text{span}\{\,e^{-it\widehat H_{\mathcal G}}|\psi_0\rangle:\ t\in\mathbb R\,\}
\subseteq \mathcal H_0 .
\end{equation}
This subspace contains the full trajectory of $|\psi_0\rangle$ and is preserved by the time evolution. Therefore, if $\widehat{P}_{\mathcal K}$ denotes the orthogonal projector onto $\mathcal K(\psi_0)$, then inserting $\widehat{P}_{\mathcal K}$ on both sides of the observable does not change the signal:
\begin{multline}
f_{S,S'}(t)
=
\\
\big\langle \psi_0\big|
e^{it\widehat H_{\mathcal G}}
\big[\widehat{P}_{\mathcal K}(\widehat O_S-\widehat O_{S'})\widehat{P}_{\mathcal K}\big]
e^{-it\widehat H_{\mathcal G}}
\big|\psi_0\big\rangle
\end{multline}
Thus, for this particular initialization, distinguishability is governed by the compressed observable
\begin{equation}
\widehat\Delta_{\mathcal K}
:=
\widehat{P}_{\mathcal K}(\widehat O_S-\widehat O_{S'})\widehat{P}_{\mathcal K}.
\end{equation}
The possible loss of distinguishability is therefore an accessibility issue: the state probes only the restriction of $\widehat O_S-\widehat O_{S'}$ to $\mathcal K(\psi_0)$. 

This perspective also makes clear the role of the transverse-field term of the Hamiltonian. For the implementation-relevant initialization $|\psi_0\rangle=|0\rangle^{\otimes N}$, when $\Omega=0$, the Ising Hamiltonian is diagonal in the computational basis, so $|0\rangle^{\otimes N}$ is stationary and
\begin{equation}
\mathcal K(\psi_0)=\text{span}\{|\psi_0\rangle\}.
\end{equation}
Consequently, for nonempty $S,S'\subseteq\mathcal V$, the compressed signal satisfies $f_{S,S'}\equiv 0$. By contrast, in the regime where $\Omega\neq 0$, the transverse-field term couples $|0\rangle^{\otimes N}$ to one-excitation states and can generate a nontrivial Krylov subspace inside $\mathcal H_0$. Section~\ref{subsec:empirical_validation} provides an instance in this regime as used in the experiments.
\end{remark}

\section{Proof of Corollary~\ref{cor:corolary_main}}
\label{appx:cor}
\begin{proof}
Let
\begin{equation}
\mathcal G_{1,2}
=
\mathcal G_1\cup\mathcal G_2,
\qquad
\Gamma_{1,2}
=
\operatorname{Aut}(\mathcal G_{1,2}).
\end{equation}
Since \(\mathcal G_1\) and \(\mathcal G_2\) are connected and have disjoint vertex sets, the connected components of \(\mathcal G_{1,2}\) are precisely \(\mathcal V_1\) and \(\mathcal V_2\). We first prove the reverse implication. Assume that there exists an orbit of \(\Gamma_{1,2}\) containing vertices from both components. Then there are \(u\in\mathcal V_1\), \(v\in\mathcal V_2\), and \(g\in\Gamma_{1,2}\) such that
$g\star u=v.$
\newline 
Because \(g\) is an automorphism, it preserves adjacency and hence preserves connectedness. Therefore, $g(\mathcal V_1)$ is a connected component of $\mathcal G_{1,2}$. Since \(u\in\mathcal V_1\) and \(g\star u=v\in\mathcal V_2\), the image \(g(\mathcal V_1)\) intersects \(\mathcal V_2\). As \(\mathcal V_2\) is a connected component, this forces $g(\mathcal V_1)=\mathcal V_2.$
\newline The restriction $ g|_{\mathcal V_1}:\mathcal V_1\to\mathcal V_2 $
is therefore a bijection. Since \(g\) preserves adjacency and non-adjacency in \(\mathcal G_{1,2}\), this
restriction is an isomorphism from \(\mathcal G_1\) to \(\mathcal G_2\). Hence $\mathcal G_1\simeq \mathcal G_2.$
\newline We now prove the forward implication. Assume that $ \mathcal G_1\simeq\mathcal G_2,$
and let $ \phi:\mathcal V_1\to\mathcal V_2 $ be a graph isomorphism. Define a permutation \(\sigma\) of \(\mathcal V_{1,2}:=\mathcal V_1\sqcup\mathcal V_2\) by
\begin{equation}
\sigma(w) =
\begin{cases}
\phi(w), & w\in\mathcal V_1,\\
\phi^{-1}(w), & w\in\mathcal V_2.
\end{cases}
\end{equation}
We claim that \(\sigma\in\Gamma_{1,2}\). Indeed, if \((u,w)\in\mathcal E_1\), then
\begin{equation}
(\sigma(u),\sigma(w))=(\phi(u),\phi(w))\in\mathcal E_2
\end{equation}
because \(\phi\) is an isomorphism. Conversely, if \((u,w)\in\mathcal E_2\), then
\begin{equation}
(\sigma(u),\sigma(w))=(\phi^{-1}(u),\phi^{-1}(w))\in\mathcal E_1.
\end{equation}
Moreover, there are no edges between \(\mathcal V_1\) and \(\mathcal V_2\) in the disjoint union \(\mathcal G_{1,2}\), and \(\sigma\) maps each component onto the other. Hence \(\sigma\) preserves the edge set of \(\mathcal G_{1,2}\), so $\sigma\in\Gamma_{1,2}.$
Now take any \(u\in\mathcal V_1\). Since \(\sigma\in\Gamma_{1,2}\), the orbit of \(u\) under \(\Gamma_{1,2}\) contains both
$
u\in\mathcal V_1
\text{ and }
\sigma(u)=\phi(u)\in\mathcal V_2.
$
Thus there exists an orbit of \(\Gamma_{1,2}\) containing at least one vertex from each connected component. This proves the equivalence.
\end{proof}
\section{Proof of Proposition~\ref{thm:stat_test}}
\label{app:appendix_stat_test}
\begin{proof} Fix a time \(t_0>0\) such that
\begin{equation}
\gap(t_0) = \min_{u\in\mathcal V_1,\;v\in\mathcal V_2} \left| \langle \widehat n_u(t_0)\rangle - \langle \widehat n_v(t_0)\rangle \right| >0.
\end{equation}

The positivity of this quantity holds for Haar-almost every product initial state and for almost every \(t_0>0\), as explained in Proposition~\ref{thm:stat_test}.

For each vertex \(w\in\mathcal V_1\sqcup\mathcal V_2\), let
\begin{equation}
p_w(t_0):=\langle \widehat n_w(t_0)\rangle.
\end{equation}

A projective measurement of \(\widehat n_w\) at time \(t_0\) produces a Bernoulli random variable with mean \(p_w(t_0)\). Let \(X_{w,m}(t_0)\in\{0,1\}\) denote the \(m\)-th independent measurement sample of \(\widehat n_w\), and define the empirical estimator
\begin{equation}
\bar n_w(t_0)
:=
\frac{1}{M}\sum_{m=1}^{M}X_{w,m}(t_0).
\end{equation}

By Hoeffding's inequality, for every \(\epsilon>0\),
\begin{equation}
\Pr\!\left(
\left|\bar n_w(t_0)-p_w(t_0)\right|\geq \epsilon
\right)
\leq
2e^{-2M\epsilon^2}.
\end{equation}

Taking a union bound over all $N=|\mathcal V_1|+|\mathcal V_2|$ vertices gives
\begin{multline}
\Pr\!\left(
\left|\bar n_w(t_0)-p_w(t_0)\right|<\epsilon
\right)
\geq
\\
1-2Ne^{-2M\epsilon^2} \;\; \forall w\in\mathcal V_1\sqcup\mathcal V_2.
\end{multline}
Thus, choosing $\epsilon=\sqrt{\frac{1}{2M}\log\!\left(\frac{2N}{\alpha}\right)}$
ensures that all occupation estimates are simultaneously within \(\epsilon\) of their expectations with probability at least \(1-\alpha\). On this event, define for each vertex \(w\in\mathcal V_1\sqcup\mathcal V_2\) the confidence interval
\begin{equation}
J_w
:=
\left[
\bar n_w(t_0)-\epsilon,\,
\bar n_w(t_0)+\epsilon
\right]
\cap [0,1].
\end{equation}

Since \(|\bar n_w(t_0)-p_w(t_0)|<\epsilon\), we have
$ p_w(t_0)\in J_w \text{ and } J_w \subseteq \left[ p_w(t_0)-2\epsilon,\, p_w(t_0)+2\epsilon \right].$
Now consider any cross-component pair \(u\in\mathcal V_1\) and \(v\in\mathcal V_2\). By definition of
$\gap(t_0)$ : $ |p_u(t_0)-p_v(t_0)| \geq \gap(t_0).$ Assume without loss of generality that \(p_u(t_0)>p_v(t_0)\). \newline Then
\begin{equation}
\inf J_u
\geq
p_u(t_0)-2\epsilon \;\;,
\;\;
\sup J_v
\leq
p_v(t_0)+2\epsilon.
\end{equation}

Therefore:
\begin{equation}
\inf J_u-\sup J_v
\geq
p_u(t_0)-p_v(t_0)-4\epsilon
\geq
\gap(t_0)-4\epsilon.
\end{equation}

Hence, if
$
\epsilon<\frac{\gap(t_0)}{4},
$
then \(J_u\cap J_v=\emptyset\) for every cross-component pair
\(u\in\mathcal V_1\), \(v\in\mathcal V_2\). In other words, with probability at least
\(1-\alpha\), the empirical confidence intervals certify that no vertex occupation value in
\(\mathcal V_1\) coincides with any vertex occupation value in \(\mathcal V_2\).

Using
$
\epsilon
=
\sqrt{
\frac{1}{2M}
\log\!\left(\frac{2N}{\alpha}\right)
},
$
the condition \(\epsilon<\gap(t_0)/4\) is satisfied whenever
$
M
>
\frac{8}{\gap(t_0)^2}
\log\!\left(\frac{2N}{\alpha}\right).
$
This proves the claimed shot-complexity scaling
$
M
=
O\!\left(
\gap(t_0)^{-2}
\log\!\left(\frac{N}{\alpha}\right)
\right).
$
\end{proof}
\section{Proof of Theorem~\ref{thm:expressiveness}}
\label{app:thm_expressiveness}
\begin{proof}
Throughout the proof, the Hamiltonian used to generate the dynamical features is fixed. All genericity statements are therefore with respect to the choice of initial states and sampling times only, and not with respect to the Hamiltonian parameters.

Let $\mathcal D=\{\mathcal G_1,\dots,\mathcal G_n\}$
be a finite dataset of graphs, padded to a common size \(N\). We write
\begin{equation}
\mathcal P:=\{(k,l):1\leq k<l\leq N\}
\end{equation}
for the set of unordered node pairs. For each graph \(\mathcal G=(\mathcal V,\mathcal E)\), the quantum-dynamical preprocessing produces node time-series features
\begin{equation}
\mathbf x_i(\mathcal G)\in\mathbb R^T,
\qquad
[\mathbf x_i(\mathcal G)]_t
=
\langle \widehat n_i(t)\rangle_{\mathcal G},
\end{equation}
and pair time-series features $\mathbf e_{ij}(\mathcal G)\in\mathbb R^T$:
\begin{multline}
[\mathbf e_{ij}(\mathcal G)]_t
=
\langle \widehat n_i\widehat n_j(t)\rangle_{\mathcal G}
-
\langle \widehat n_i(t)\rangle_{\mathcal G}
\langle \widehat n_j(t)\rangle_{\mathcal G},
\\ i<j.
\end{multline}
We stack these features into tensors
\begin{equation}
X_{\mathrm{in}}(\mathcal G)\in\mathbb R^{N\times T},
\qquad
[X_{\mathrm{in}}(\mathcal G)]_{i,:}
=
\mathbf x_i(\mathcal G),
\end{equation}
and
\begin{equation}
E_{\mathrm{in}}(\mathcal G)\in\mathbb R^{|\mathcal P|\times T},
\qquad
[E_{\mathrm{in}}(\mathcal G)]_{(i,j),:}
=
\mathbf e_{ij}(\mathcal G).
\end{equation}

We next use the orbit-distinguishability consequence of Theorem~\ref{thm:main} together with Corollary~\ref{cor:corolary_main}. Namely, for any non-isomorphic graph pair
$
\mathcal G_k\not\simeq \mathcal G_l
$
in \(\mathcal D\), and for any cross-component node pair
$
i\in\mathcal V_k,
\;
j\in\mathcal V_l,
$
the corresponding one-body occupation trajectories are generically distinct, meaning 
$
t
\longmapsto
\langle \widehat n_i(t)\rangle_{\mathcal G_k}
-
\langle \widehat n_j(t)\rangle_{\mathcal G_l}
$
is not identically zero for Haar-almost every choice of initial states in the relevant trivial symmetry sectors, and its zero set has Lebesgue measure zero.

Since \(\mathcal D\) is finite and each graph has finitely many vertices, a finite-intersection argument gives that, for Haar-almost every choice of initial states and for almost every choice of sampling times, this separation holds simultaneously for all non-isomorphic graph pairs \((k,l)\) and all cross-component node pairs \((i,j)\). Hence, whenever
$
\mathcal G_k\not\simeq \mathcal G_l,
$
we have
\begin{equation}
\mathbf x_i(\mathcal G_k)
\neq
\mathbf x_j(\mathcal G_l)
\qquad
\forall i\in\mathcal V_k,\; j\in\mathcal V_l.
\end{equation}
Consequently, no relabeling of the padded vertex set can make the node-level dynamical feature tensors coincide, and therefore
\begin{equation}
\min_{\pi\in S_N}
\left\|
X_{\mathrm{in}}(\mathcal G_k)
-
P_\pi X_{\mathrm{in}}(\mathcal G_l)
\right\|_F^2
>0.
\end{equation}

We now show that our model contains a parameter setting that preserves these separating features. Assume that the output of $L$ stacked \QDAGer layers is (for either of the graphs from the input pair) : 
\begin{equation}
\mathcal E_\theta(\mathcal G)
=
\bigl(X_\theta(\mathcal G),R_\theta(\mathcal G)\bigr),
\end{equation}
with
\begin{equation}
X_\theta(\mathcal G)\in\mathbb R^{N\times D},
\qquad
R_\theta(\mathcal G)\in\mathbb R^{|\mathcal P|\times D}.
\end{equation}
Without loss of generality, we choose the embedding dimension so that $D=T$, the dimension of the dynamical time series introduced in section~\ref{subsec:model_archi}. Let
\(X_{\mathrm{in}}(\mathcal G)\in\mathbb R^{N\times T}\) denote the matrix whose $i$-th row is the node feature
\(\mathbf x_i\in\mathbb R^T\), and let the initial pair features be the vectors
\(\mathbf e_{ij}\in\mathbb R^T\) defined in the paragraph on quantum dynamical features. We claim that the class of \QDAGer layers covers, as a special parameter setting, the map that leaves all node features unchanged and discards the pair features. More precisely, there exists \(\theta_0\) such that, for every \(\mathcal G\in\mathcal D\),
\begin{equation}
X_{\theta_0}(\mathcal G)
=
X_{\mathrm{in}}(\mathcal G),
\qquad
R_{\theta_0}(\mathcal G)
=
0.
\end{equation}
Indeed, consider one layer and use the notation of section~\ref{subsec:model_archi}. Set the value projection and the node--pair aggregation map in \texttt{NodeUpdate} to zero,
\begin{equation}
W_V=0,
\qquad
W_{NE}=0.
\end{equation}
Then, independently of the attention coefficients \(\alpha_i^{(g)}\) and \(\alpha_{ij}^{(g)}\), the node update term defined in section~\ref{subsec:model_archi} vanishes:
\begin{equation}
\mathbf x_i^{\prime(g)}
=
\alpha_i^{(g)} W_V\mathbf x_i^{(g)}
+
\sum_{j\neq i}\alpha_{ij}^{(g)} W_{NE}\mathbf e_{ij}^{\prime(g)}
=0.
\end{equation}
Thus the residual part of the node block contributes exactly
\(\mathbf x_i^{(g)}+\mathbf x_i^{\prime(g)}=\mathbf x_i^{(g)}\). Next, set the pair branch to the zero map by choosing, in the notation of \texttt{PairUpdate},
\begin{equation}
W_Q=W_K=W_{Ew}=W_{Eb}=0
\end{equation}
and take any optional pair residual or output projection to be zero. With this choice, the pair branch does not contribute information to the output pair embedding, so the final pair embedding can be set to zero. The query and key projections are therefore harmless: after this choice, the attention logits may still be formally defined, but their values cannot affect the node update as both quantities multiplied by the attention weights in Eq.~\eqref{eq:node_update} are already zero.

It remains only to specify the node feed-forward block. Since the theorem concerns representational capacity, we may take the normalization block to be absent or inactive in this construction, so that it does not modify the preserved coordinates. We then choose the node feed-forward block to be the identity map on \(\mathbb R^T\). A two-layer ReLU network realizes the identity exactly: for any \(\mathbf x\in\mathbb R^T\), set
\begin{align}
W_1 &= \begin{bmatrix} I_T \\ -I_T \end{bmatrix}, & b_1 &= 0, \\
W_2 &= \begin{bmatrix} I_T & -I_T \end{bmatrix}, & b_2 &= 0.
\end{align}
Then
\begin{equation}
W_2\operatorname{ReLU}(W_1\mathbf x)+b_2
=
\operatorname{ReLU}(\mathbf x)-\operatorname{ReLU}(-\mathbf x)
=\mathbf x,
\end{equation}
coordinatewise, since \(\max\{x_t,0\}-\max\{-x_t,0\}=x_t\) for every coordinate \(t=1,\dots,T\). Applying the same parameter setting in each of the \(L\) stacked layers proves the claimed preservation of node-level dynamical features and annihilation of pair-level features.

For this parameter setting, define the embedding collision distance
\begin{multline}
\varepsilon_{\theta_0}(\mathcal G_k,\mathcal G_l)
:=
\min_{\pi\in S_N}
\Bigg(
\left\|
X_{\theta_0}(\mathcal G_k)
-
P_\pi X_{\theta_0}(\mathcal G_l)
\right\|_F^2
\\
+
\left\|
R_{\theta_0}(\mathcal G_k)
-
S_\pi R_{\theta_0}(\mathcal G_l)
\right\|_F^2
\Bigg)^{1/2},
\end{multline}
where \(S_\pi\) is the permutation matrix induced by \(\pi\) on \(\mathcal P\). Since \(R_{\theta_0}(\mathcal G)=0\) for every \(\mathcal G\), this reduces to
$
\varepsilon_{\theta_0}(\mathcal G_k,\mathcal G_l)
=
\min_{\pi\in S_N}
\left\|
X_{\mathrm{in}}(\mathcal G_k)
-
P_\pi X_{\mathrm{in}}(\mathcal G_l)
\right\|_F.
$
Therefore, for every non-isomorphic pair \(\mathcal G_k\not\simeq \mathcal G_l\), the feature-separation property gives
$
\varepsilon_{\theta_0}(\mathcal G_k,\mathcal G_l)>0.
$
Thus the parameter setting \(\theta_0\) yields distinct \QDAGer embeddings, up to node relabeling, for all non-isomorphic graph pairs in \(\mathcal D\). This proves the claim.
\end{proof}

\section{Experiments Details}
\label{appx:expe_details}
\subsection{Efficient evaluation of the permutation-invariant separation index}
\label{appx:expe_details_1}
The separation index $\eta^{(1)}_{ij}$ (\textit{cf.} Eq.~\eqref{eq:RMS}) between two graphs $\mathcal G_i$ and
$\mathcal G_j$ of a family involves a minimization over the symmetric group $S_{N}$ of all relabelings of the $N$ vertices. This is never carried out by an explicit search over its $N!$ elements. Because the objective is additive over vertices (the total cost is a sum of contributions, each depending only on a single vertex $u$ and its image $\pi(u)$) the minimization is exactly a linear assignment (minimum-cost bipartite matching) problem. Collecting the per-vertex time-averaged mismatches $m_{uv}$ of Eq.~\eqref{eq:cost_matrix} into a cost matrix, the optimal relabeling is the permutation minimizing $\sum_u m_{u\pi(u)}$, which is solved exactly by the Hungarian (Kuhn--Munkres) algorithm in $\mathcal{O}(N^3)$ time. We use the \texttt{linear\_sum\_assignment} routine of SciPy, which implements a modified Jonker--Volgenant algorithm with the same $\mathcal{O}(N^3)$ worst-case complexity. The apparent $N!$ complexity of optimizing over $S_{N}$ therefore collapses to a low-degree polynomial cost, so that even the larger families considered here are handled in a reasonably short time.
\subsection{Classical features for ablation}
\label{appx:expe_details_2}

To isolate the representational contribution of the quantum-dynamical signal, we conduct an ablation study where the quantum dynamics are replaced by standard classical structural descriptors: relative Random-Walks (RWs) and Heat Kernel (HK) signatures. For a graph $\mathcal{G}$ with adjacency matrix $A$ and diagonal degree matrix $D$ (where $D_{ii} = \sum_j A_{ij}$), the $t$-step random walk transition probability matrix is given by:

\begin{equation}
    P^{(t)} = (D^{-1}A)^t
\end{equation}
Similarly, the Heat Kernel features are governed by the unnormalized graph Laplacian $L = D - A$. The heat kernel matrix at diffusion time $t$ is defined by the matrix exponential, which can be written its using its eigenpairs $(\lambda_m, u_m)$:
\begin{equation}
    K_{heat}^{(t)} = e^{-tL} = \sum_{k=1}^N e^{-t\lambda_k} u_k u_k^\top
\end{equation}
To integrate these classical signals into the \QDAGer architecture without modifying the downstream attention or loss mechanisms, we evaluate these matrices over a set of $T$ discrete steps (either discrete walk steps or continuous diffusion times). Denoting the resulting sequence of dense matrices as $M^{(t)}$ (where $M \in \{P, K_{heat}\}$), we assign the classical auto-correlations (the diagonal entries) to the node features $\mathbf{x}_i \in \mathbb{R}^T$, such that $[\mathbf{x}_i]_t = M^{(t)}_{i,i}$. The cross-correlations (the off-diagonal entries) replace the connected two-point quantum correlators and are mapped to the pair features $\mathbf{e}_{ij} \in \mathbb{R}^T$, where $[\mathbf{e}_{ij}]_t = M^{(t)}_{i,j}$ for $i \neq j$.
\subsection{Model hyperparameters}
\label{appx:expe_details_3}
 We put in Table~\ref{tab:qdager_equal_core} the configurations for the different models and datasets for which results are shown in Table~\ref{tab:res_table}

\begin{table*}[htbp]
\centering
\scriptsize
\setlength{\tabcolsep}{3pt}
\caption{Each entry is written as dir/SH, where the left values describe the configurations of the models trained with the blindfold (dir) loss, and the right value those from the Sinkhorn (SH) surrogate loss. A single value is shown wherever the configuration is the same between the two variants. A dash (\texttt{-}) means that the row is not applicable for that variant (for example, dir has no Sinkhorn parameters). The most systematic differences are the normalization and regularization choices (dir: batch norm + nonzero encoder/transformer dropout, SH : layer norm + zero encoder/transformer dropout). T/F denotes booleans.}
\label{tab:qdager_equal_core}
\resizebox{\textwidth}{!}{%
\begin{tabular}{lccccccc}
\toprule
Row & AIDS & Yeast & Mutag & MolHIV & MolPCBA & Code2 & Linux \\
\midrule
Size (num pairs) & 182820 & 220500 & 117166 & 220500 & 220500 & 3629 & 1774 \\
N (max size) & 20 & 20 & 20 & 20 & 20 & 20 & 20 \\
Batch size & 1024/512 & 1024/512 & 1024/512 & 1024/512 & 1024/512 & 256 & 256 \\
Enc in dim & 100 & 100 & 100 & 100 & 100 & 32 & 32 \\
Enc hid dim & 64 & 64 & 64 & 64 & 64 & 16 & 16 \\
Enc out dim & 32 & 32 & 32 & 32 & 32 & 16 & 16 \\
Enc num layers & 2 & 2 & 2 & 2 & 2 & 2 & 2 \\
Enc dropout & 0.1/0 & 0.1/0 & 0.1/0 & 0.1/0 & 0.1/0 & 0.1/0 & 0.1/0 \\
Tr num layers & 5 & 5 & 4/5 & 5 & 5 & 4 & 4 \\
Tr layer dim & 32 & 32 & 32 & 32 & 32 & 16 & 16 \\
Num heads & 4 & 4 & 2 & 4 & 4 & 2 & 2 \\
Tr dropout & 0.1/0 & 0.1/0 & 0.1/0 & 0.1/0 & 0.1/0 & 0.1/0 & 0.1/0 \\
Attn dropout & 0.1 & 0.1 & 0.1 & 0.1 & 0.1 & 0.1 & 0.1 \\
Layer norm & F/T & F/T & F/T & F/T & F/T & F/T & F/T \\
Batch norm & T/F & T/F & T/F & T/F & T/F & T/F & T/F \\
Learning rate  & 1e-3 & 1e-3 & 1e-3 & 1e-3 & 1e-3 & 1e-3 & 1e-3 \\
Weight decay & 1e-5 & 1e-5 & 1e-5 & 1e-5 & 1e-5 & 1e-5 & 1e-5 \\
Num epochs & 1500 & 1500 & 1500 & 1500 & 1500 & 1500 & 1500 \\
Warmup & 50 & 50 & 50 & 50 & 50 & 50 & 50 \\
SH use & F/T & F/T & F/T & F/T & F/T & F/T & F/T \\
SH hid dim & -/16 & -/16 & -/32 & -/16 & -/16 & -/16 & -/16 \\
SH out dim & -/16 & -/16 & -/16 & -/16 & -/16 & -/8 & -/8 \\
SH tau & -/0.1 & -/0.1 & -/0.05 & -/0.1 & -/0.1 & -/0.1 & -/0.1 \\
SH $N_{\text{iters}}$ & -/20 & -/20 & -/20 & -/20 & -/20 & -/20 & -/20 \\
\bottomrule
\end{tabular}%
}
\end{table*}

\clearpage

\end{document}